\documentclass[runningheads,a4paper]{llncs}
\pdfoutput=1
\usepackage{enumitem}
\usepackage{graphicx}
\usepackage{amsmath}
\usepackage{amsfonts}
\usepackage{amssymb}
\usepackage{verbatim}
\usepackage{amssymb}
\usepackage{mathtools}
\usepackage{graphicx}
\usepackage{comment}
\usepackage{color}

\input{Qcircuit.tex}
\def\beq{\begin{equation}}
\def\eeq{\end{equation}}
\def\beqa{\begin{eqnarray}}
\def\eeqa{\end{eqnarray}}

\usepackage{url}
\urldef{\mailsa}\path|{alfred.hofmann, ursula.barth, ingrid.haas, frank.holzwarth,|
\urldef{\mailsb}\path|anna.kramer, leonie.kunz, christine.reiss, nicole.sator,|
\urldef{\mailsc}\path|erika.siebert-cole, peter.strasser, lncs}@springer.com|
\newcommand{\keywords}[1]{\par\addvspace\baselineskip
\noindent\keywordname\enspace\ignorespaces#1}

\usepackage{algorithm}
\usepackage{algorithmic}
\floatname{algorithm}{Protocol}

\begin{document}

\mainmatter  

\title{Composable secure multi-client delegated quantum computation}

\titlerunning{Composable secure multi-client delegated quantum computation}

\author{Monireh Houshmand$^{1,2}$
\and Mahboobeh Houshmand$^{1,2}$\and Si-Hui Tan$^{1,2}$\and Joseph F. Fitzsimons$^{1,2}$
}

\authorrunning{Composable secure multi-client delegated quantum computation}

\institute{$^{1}$Singapore University of Technology and Design, 8 Somapah Road, Singapore 487372,\\
$^2$Centre for Quantum Technologies, National University of Singapore, Block S15, 3 Science Drive 2, Singapore 117543
}

\toctitle{}
\tocauthor{}
\maketitle

\begin{abstract}
The engineering challenges involved in building large scale quantum computers, and the associated infrastructure requirements, mean that when such devices become available it is likely that this will be in limited numbers and in limited geographic locations. It is likely that many users will need to rely on remote access to delegate their computation to the available hardware. In such a scenario, the privacy and reliability of the delegated computations are important concerns. On the other hand, the distributed nature of modern computations has led to a widespread class of applications in which a group of parties attempt to perform a joint task over their inputs, e.g., in cloud computing. In this paper, we study the multi-client delegated quantum computation problem where we consider the global computation be made up of local computations that are individually decided by the clients. Each client part is kept secret from the server and the other clients. We construct a composable secure multi-client delegated quantum computation scheme from any composable secure single-client delegated quantum computation protocol and quantum authentication codes.
\end{abstract}

\section{Introduction}

The distributed nature of modern computations has led to a widespread class of applications in which a group of parties perform a joint task over their inputs. In these situations, it is very often desirable that the parties' privacy should be preserved due to the fact that they may be mutually distrustful. This has stimulated extensive research in classical cryptography, in the form of secure multi-party computation (MPC)~\cite{yao1982protocols,goldwasser1997multi,pinkas2009secure}. The aim of MPC is to guarantee that the joint public task is carried out correctly and that the parties will not obtain any information from their interactions other than the final output of the computation and what is naturally leaked from it. Some applications of MPC include secure decentralized elections, secure auctions and private data mining. Private function evaluation~\cite{mohassel2014actively,mohassel2013hide,barni2009secure,canetti2001selective} (PFE), is a special case of MPC, where the parties compute a function which is determined privately by one of the parties. The key additional security requirement is that the only permitted leakage to an adversary about the function is the size of the circuit evaluating it.

There has also been work on quantum extensions of MPC, namely secure multi-party quantum computation (MPQC)~\cite{crepeau2002secure,ben2006secure}, based on verifiable quantum secret sharing.
In~\cite{dupuis2010secure}, the problem of privately evaluating some unitary transformation against the quantum version of classical semihonest adversaries over a joint input state held by two parties was addressed. This study was later made secure against more malicious adversaries in~\cite{dupuis2012actively}.

In the schemes mentioned above, all parties are assumed to have quantum computational power. However, as realising scalable and reliable quantum computers is an important challenge, when they can be built, it is likely that there will be a limited number of them available. In this case, clients with restricted quantum capabilities may need to securely delegate their computations to a universal quantum server, in the sense that their inputs, computation and outputs remain hidden from the server. This is known as blind quantum computation (BQC).

BQC was first considered by Childs~\cite{childs2005secure} using the idea of encrypting the input qubits with a quantum one-time pad~\cite{ambainis2000private}. A subsequent research in this field was presented in~\cite{arrighi2006blind}, which was equipped with the possibility of detecting a malicious server. However, these approaches required complex quantum capabilities on the client's side. The complexity of the client's device was reduced in~\cite{broadbent2009universal} to preparing a certain set of single qubit states, and a universal blind quantum computing (UBQC) protocol was presented. This was the first protocol to take advantage of the measurement-based quantum computation model (MBQC)~\cite{raussendorf2001one} for BQC. A different approach for utilising MBQC for UBQC was proposed in~\cite{morimae2013blind} where the client only makes measurements instead of state preparation. The security of this protocol is based on the no-signaling principle~\cite{popescu1994quantum}.
In recent years, various aspects of BQC have been studied \cite{sueki2013ancilla,morimae2012continuous,morimae2015ground,kashefi2016extending,morimae2014verification,dunjko2012blind,morimae2013blind,dunjko2016blind,mantri2016flow,giovannetti2013efficient,mantri2013optimal,perez2015iterated,morimae2011blind,morimae2011blind}, and a number of schemes have been successfully implemented in experiment \cite{barz2012demonstration,barz2013experimental,greganti2016demonstration,huang2017experimental}.

A deeply coupled property with any BQC is verifiability, which is the ability of the client to verify the correctness of delegated quantum computation (DQC). The verification mechanism allows clients to either confirm that the server has followed the instructions of the protocol, or detect if it has deviated from the protocol. There are basically two main approaches for the verification of BQC in the literature, namely trapped-based and stabilizer testing. The former approach typified by~\cite{fitzsimons2012unconditionally} uses hidden traps in the resource state that do not affect the computation. Following that, trap-based techniques have been utilised in~\cite{morimae2014verification,gheorghiu2015robustness,hajduvsek2015device}. The second approach~\cite{morimae2013blind} only requires the server to generate many-qubit states.
Thus, the client only has to verify those states and this goal can be achieved by directly verifying the stabilizers of the states, which is called stabilizer testing~\cite{hayashi2015verifiable,morimae2016measurement}.

However, in the results described above, the security definition under which security of the protocols has been considered was in a stand-alone setting, where the protocol is studied in isolation from the environment. For instance, in~\cite{broadbent2009universal} a protocol is said to be blind if the distribution of the classical and the quantum information obtained by the server is fully determined by the permitted leakage. The permitted leakage is what the server can unavoidably learn from the interactions with the client. Most of the BQC protocols proposed so far have used some form of this definition.
However, the stand-alone definition of blindness does not provide any guarantee on the behaviour of the protocol when it is used as a subroutine in a larger system. A stronger form of cryptographic security is composable security~\cite{canetti2001universally,pfitzmann2001model,maurer2011abstract} which certifies the security of a protocol when run together with many other protocols, even if they are insecure. Specifically, composable security definitions for DQC were introduced in~\cite{dunjko2014composable}. It has been shown that the UBQC protocols in~\cite{fitzsimons2012unconditionally} and~\cite{morimae2013blind} satisfy these definitions.

The problem of secure delegation of a multi-party quantum computation to a server having a universal quantum computer is addressed in~\cite{kashefi2016blind}.
In this setting, a number of computationally weak clients (possibly mistrustful) intend to perform a global unitary operation, known to all of them a priori, on their joint input state with the help of a powerful untrusted server.
In this approach, the private data of the clients remain secret during the protocol and the computation performed is not known to the server.
The security of the protocol against a dishonest server and against dishonest clients, under the asumption of common classical cryptographic constructions, is proved.

In this paper, we examine the multi-client DQC problem from a different perspective. In our setting, we let the global unitary be made up of local unitaries that are individually decided by the clients. All local unitaries of each client are kept secret from the server and the other clients. We construct a composable secure multi-client DQC scheme from any composable secure single-client DQC and quantum authentication codes. Our scheme finds application in, e.g., distributed cloud computing~\cite{armbrust2010view}, in which clients from multiple geographic locations process their data on a  server to perform a global computation.

The paper is organised as follows. In the next section, the needed preliminaries are provided. In Section~\ref{sec:MCDQC}, the proposed protocol for multi-client DQC and the proof of security is presented. The optimisation of the protocol is provided in Section~\ref{sec:RQC} by reducing the quantum communication between the clients and the server. Finally, Section~\ref{sec:conclu} concludes the paper.

\section{Preliminaries}

\label{sec:pre}
\subsection{Composable Security}

Local security definitions consider cryptographic protocol problems in a model where the only involved parties are the actual participants in the protocol, and only a single execution of the protocol happens. However, this model of stand-alone execution does not fully capture the security requirements of a cryptographic protocol which is embedded in a larger scheme and is run concurrently with a number of other protocols. For a cryptographic protocol to be usable in a larger context, the security definitions need to be composable.

A strong and desirable security definition is given by composable security, using three approaches: the universal composability (UC) framework~\cite{canetti2001universally}, reactive simulatability (RS)~\cite{pfitzmann2001model} and the abstract cryptography (AC) developed by Maurer and Renner ~\cite{maurer2011abstract}. These security definitions have also been extended to the quantum setting~\cite{ben2004general,unruh2004simulatable,portmann2015causal}.

In contrast to UC and RS, AC uses a top-down paradigm which first considers a higher level of abstraction, provides definitions and theorems at this level and then proceeds downwards, introducing in each new lower level only the absolutely essential specializations. The fundamental element of AC is a \textit{system}, an abstract object with interfaces through which it interacts with the environment  and/or other systems. Systems are composed of three classes, namely resources, converters, and distinguishers. A \textit{resource} is a system with a set of interfaces. Generally, each interface is devoted to one party, representing the functionalities that are accessible to that party. \textit{Converters} are used to transform one resource into another and capture operations that a player might perform locally at their interface. These are systems with only two interfaces, one is designed as an inside interface and the other one as an outside interface.
The inside interface is connected to the resources and the outside interface receives inputs and produces outputs.
We write either $\pi_i\mathcal{R}$ or $\mathcal{R}\pi_i$ to describe the new resource with the converter $\pi_i$ connected at the interface $i$.

A set of converters which describe how available resources are combined to realise another resource is called a \textit{protocol}. Another type of converter that we need is a \textit{filter}.
For a cryptographic protocol to be useful, it
is not sufficient to provide guarantees on what happens when an adversary
is present, one also has to provide a guarantee on what happens when no
adversary is present. We model this setting by covering the adversarial interface with a filter
that emulates honest behaviour.

The closeness of two resources is defined by a pseudo-metric $d(\cdot,\cdot)$ on the space of resources which is captured by a \textit{distinguisher}. It is a system which having access to all interfaces of one of the two resources, outputs a bit, deciding which resource it is interacting with. The \textit{distinguisher's advantage} is defined as the difference of probability of the correct output when the distinguisher is connected to either of resources. Since the distinguishing advantage is a pseudo-metric, it respects the triangle inequality. For any resources $\mathcal{R}, \mathcal{S},$ and $\mathcal{T},$ we have
\begin{equation}
 d(\mathcal{R},\mathcal{S}) \leq d(\mathcal{R},\mathcal{T}) + d(\mathcal{T},\mathcal{S}). \label{eq:triangle-ineq}
 \end{equation}
Distinguishing advantage is non-increasing under composition with any other system. Therefore for any converter $\alpha$ and resources $\mathcal{R}, \mathcal{S}$ and $\mathcal{T},$
\begin{equation}
d(\alpha \mathcal{R},\alpha \mathcal{S}) \leq d(\mathcal{R},\mathcal{S}), \label{eq:non-inc-conv}
\end{equation}
\begin{equation}
d(\mathcal{R}\mathcal{T},\mathcal{S}\mathcal{T}) \leq d(\mathcal{R},\mathcal{S}). \label{eq:non-inc-res}
\end{equation}

Two resources $\mathcal{R}$ and $\mathcal{S}$ are considered $\epsilon\text{-}$close which is denoted by $\mathcal{R} \underset{\epsilon}{\approx}\mathcal{S}$ if no distinguisher (in a certain class of distinguishers) has a distinguishing advantage of more than $\epsilon.$

AC views cryptography as a resource theory: a protocol constructs
a (strong) resource given some (weak) resources. Due to this constructive
aspect of the framework, it is also called \textit{constructive cryptography}
in the literature~\cite{maurer2011constructive,maurer2016indifferentiability}. If a protocol $\pi$ constructs some resource $\epsilon$-close to $\mathcal{S}$ from another resource $\mathcal{R},$ we write:
\begin{equation}
\mathcal{R} \xrightarrow[]{\pi,\epsilon} \mathcal{S}.
\label{Eq:protocol}
\end{equation}

For a cryptographic construction to be usable in any possibly complex manner, the following composable conditions must be held:

\[\mathcal{R} \xrightarrow[]{\pi,\epsilon} \mathcal{S} \;\;\text{and}\;\; \mathcal{S}\xrightarrow[]{\pi',\epsilon'} \mathcal{T} \Rightarrow \mathcal{R}\xrightarrow[]{\pi'\pi,\epsilon+\epsilon'}\mathcal{T} ,\]

\[\mathcal{R} \xrightarrow[]{\pi,\epsilon} \mathcal{S} \;\;\text{and}\;\; \mathcal{S}\xrightarrow[]{\pi',\epsilon'} \mathcal{S}' \Rightarrow \mathcal{R}||\mathcal{R'}\xrightarrow[]{\pi | \pi',\epsilon+\epsilon' }\mathcal{S}||\mathcal{S}', \]
where $\mathcal{R}||\mathcal{R'}$ is the resource resulting from parallel composition of $\mathcal{R}$ and $\mathcal{R'}$, and $\pi' \pi$ and $\pi|\pi'$ are sequential and parallel composition of protocols, respectively.

If two protocols with composable errors $\epsilon$ and $\delta$ are run jointly (e.g., one is a subroutine of the other), the error of the composed protocol is bounded by the sum of the individual errors, $\epsilon+\delta.$
In Eq. (\ref{Eq:protocol}), the resource $\mathcal{R}$ along with the protocol $\pi$ are parts of the concrete protocol and $
\mathcal{S}$ is the ideal resource we want to construct.
 Eq.~(\ref{Eq:protocol}) holds if the adversary, with access to the ideal resource, can get anything that she can get in the real world.
In other words, both systems have equivalent behaviour in the sense that they behave identically if plugged into any possible environment or application.
From the adversary's point of view, anything that could happen in the real world could identically happen in the ideal world.
Simulator systems transform the ideal resource into the real world system consisting of the real resource and the protocol.
Now we define the security of a protocol in the three-party setting with honest Alice and Bob and dishonest Eve. See~\cite{maurer2011abstract} for the general case.

\begin{definition} (Composable security~\cite{maurer2011abstract})
Let $\mathcal{R}_\phi=(\mathcal{R},\phi)$ and $\mathcal{S}_\psi=(\mathcal{R},\psi)$ be pairs of a resource ($\mathcal{R}$ and $\mathcal{S}$) with interfaces $\mathcal{I}=\{A,B,E\}$ and filters $\phi$ and $\psi$. We say that a protocol $\pi=(\pi_A,\pi_B)$ securely constructs
   $\mathcal{S}_\psi$ out of $\mathcal{R}_\phi$ within $\epsilon$ and write
\begin{equation}
\mathcal{R}_\phi \xrightarrow[]{\pi,\epsilon} \mathcal{S}_\psi,
\label{Eq:composability}
\end{equation}
 if the two following conditions hold:
\begin{itemize}
\item $d(\pi\mathcal{R}_\phi,\mathcal{S}_\psi)\leqslant\epsilon,$ and
\item there exists a converter $\sigma_E$, called a simulator, such that $d(\pi\mathcal{R},\mathcal{S}\sigma_E)\leqslant\epsilon.$
\end{itemize}
  \end{definition}
\subsection{Delegated Quantum Computation}
\label{sec:delegated}
In a two-party DQC, a client evaluates some (quantum) computation using the help of
a powerful server. The protocol must guarantee blindness, capturing the idea that the server learns nothing about the input, the output, and the computation the client evaluated other than what is unavoidable (computation size, etc.). The other security concern is verifiability, which indicates that the client can verify whether the server has followed the instructions or has deviated from them.

The composable security of DQC using the aforementioned AC framework was presented in~\cite{dunjko2014composable}.
Having introduced the notions required in the security definition, the description of the blind verifiable ideal resource model, denoted by $\mathcal{S}^{\textnormal{bv}}$, is presented as follows:
  \begin{definition}
The ideal DQC resource $\mathcal{S}^{\textnormal{bv}}$ which provides correctness, blindness and verifiability, has two interfaces, $C$ and $S$, standing for the client and the server respectively. It takes $\psi_c$, representing the description of the client's computation $U$ and quantum input $\rho_c$ at the $C$ interface and a filtered controlled bit $f^{\textnormal{bv}}$ at the $S$ interface. If the server sets $f^{\textnormal{bv}}=0,$ $\mathcal{S}^{\textnormal{bv}}$ produces the honest result $U\rho_c$ at the $C$ interface. If $f^{\textnormal{bv}}=1,$ the resource outputs the allowed leakage $\ell^{\textnormal{bv}}$ at the $S$ interface and some fixed error state $\ket{\emph{err}}\bra{\emph{err}}$ orthogonal to the output space at the $C$ interface.
 \end{definition}

Now we model the structure of a generic two-party DQC protocol, $\pi^{\textnormal{bv}}=(\pi^{\textnormal{bv}}_c,\pi^{\textnormal{bv}}_s).$ The only resource that parties require is a two-way communication channel $\mathcal{R}.$
Client's protocol $\pi_c$ receives $\psi_c$ (including the description of $\rho_c$ and $U$) as input on its outside interface. It then communicates with the server's protocol $\pi_s$ via $\mathcal{R}$ and produces some final output $\tau_c.$
Any such process where client's system is denoted $C$ and the server's system
$S$, is modelled as a sequence of CPTP maps
$\{ \mathcal{E} _i :\mathcal{L}(\mathcal{H}_{C\mathcal{R}} ) \to \mathcal{L}(\mathcal{H}_{C\mathcal{R}} )\} _{i = 1}^N$ and $\{ \mathcal{F} _i :\mathcal{L}(\mathcal{H}_{\mathcal{R}S}) \to \mathcal{L}(\mathcal{H}_{\mathcal{R}S} )\} _{i = 1}^{N-1}$
which client and the server apply sequentially to their respective systems and the communication channel $\mathcal{R}$.
The protocol $\pi^{\textnormal{bv}}=(\pi^{\textnormal{bv}}_c,\pi^{\textnormal{bv}}_s)$ is a map which, using the channel $\mathcal{R}$, transforms $\psi_c$ into $\tau_c$. If both players play honestly and the protocol is correct this should result in
$\tau_c=U\rho_c.$

\begin{definition}
A DQC protocol $\pi^{\textnormal{bv}}=(\pi^{\textnormal{bv}}_c,\pi^{\textnormal{bv}}_s)$ constructs a blind verifiable resource $ \mathcal{S}^{\textnormal{bv}}$ from the communication channel $\mathcal{R}$ within $\epsilon^{\textnormal{bv}}$ if the two following conditions are satisfied:
  \begin{equation}
 \pi_c^{\textnormal{bv}}\mathcal{R}\pi_s^{\textnormal{bv}}\underset{\epsilon^{\textnormal{bv}}}{\approx}  \mathcal{S}^{\textnormal{bv}}\perp_s,
  \label{BV-DQC1}
 \end{equation}
  \begin{equation}
 \pi_c^{\textnormal{bv}}\mathcal{R}\underset{\epsilon^{\textnormal{bv}}}{\approx} \mathcal{S}^{\textnormal{bv}}\sigma^{\textnormal{bv}}_s,
 \label{BV-DQC}
 \end{equation}
where $\perp_s$ is a protocol which just sets the filtered bit $f^{\textnormal{bv}}$ to zero, and $\sigma^{\textnormal{bv}}_s$ is some simulator.
\end{definition}
 If the first equation is satisfied, the protocol provides $\epsilon^{\textnormal{bv}}$-\emph{composable-correctness} with respect to the ideal resource,
$\mathcal{S}^{\textnormal{bv}}$. Correctness captures the notion that the input-output behaviour of the proposed protocol actually matches the input-output behaviour of the ideal resource when the server is behaving honestly.
The second condition guarantees both blindness and verifiability and if it is satisfied we say that we have $\epsilon^{\textnormal{bv}}$-\emph{blind-verifiability}. The
simulator is a protocol which, from an outside point of view, makes the ideal protocol
look like the real one, accessing only to the interface of the server.
A distinguisher which interacts with $\pi_c^{\textnormal{bv}}\mathcal{R}$ and $\mathcal{S}^{\textnormal{bv}}\sigma^{\textnormal{bv}}_s$ cannot distinguish them within $\epsilon^{\textnormal{bv}}.$ The server's part of the protocol $\pi_s^{\textnormal{bv}}$ does not appear on the left-hand side of Eq.~(\ref{BV-DQC}), which indicates that we no longer specify what the server may do. This signifies that blindness and verifiability hold for all possible cheating actions of the server. In the rest of the paper, whenever we refer to BV-DQC$(\rho, U)$, we mean a generic DQC which satisfies the conditions in Eqs.~(\ref{BV-DQC1}) and (\ref{BV-DQC}) where $\rho$ and $U$ are the client's inputs.

\subsection{Quantum Authentication Codes}
\label{sec:QAC}
A quantum authentication scheme (QAS) consists of procedures for encoding and decoding quantum information
 which allows for the detection of tampering on the encoded quantum data with a high probability. Quantum authentication codes were first
introduced by Barnum \emph{et al}~\cite{barnum2002authentication}. They also proved that quantum authentication necessarily encrypts the message as well. It therefore follows that any authentication protocol for an $m$-qubit message requires a $2m$-bit secret key $k$ to encrypt the message.
The security of quantum authentication codes is analysed in~\cite{portmann2016quantum} in the AC framework. It is proved that
the authentication protocols of~\cite{barnum2002authentication} can construct a secure quantum channel.

\begin{definition}
The secure quantum channel $\mathcal{C}^{\textnormal{q-sec}}$ (which provides both authenticity and confidentiality), shown in the top of Figure~\ref{fig:auth-ideal} (a), allows a quantum message $\rho$ to be transmitted from Alice to Bob. The resource only allows Eve to input a filtered control bit $f^\textnormal{q-sec}$ that decides if Bob receives the message ($f^\textnormal{q-sec}=0$) or $\ket{\emph{err}}\bra{\emph{err}}$ ($f^\textnormal{q-sec}=1$). It only leaks $\ell^\textnormal{q-sec}$ to the
adversary, denoted by the dashed
arrow in Figure~\ref{fig:auth-ideal} (a), which is a notification that the message has been sent as well as its length.
\end{definition}

The required resources in the concrete setting of QAS are a \emph{quantum insecure channel} $\mathcal{C}^{\textnormal{q-insec}}$ and a \emph{secret key resource} $\mathcal{K}$. $\mathcal{C}^{\textnormal{q-insec}}$ is completely under the control of the adversary. It allows the adversary to intercept the message sent from the sender to the receiver and replace it with any message of their choosing at their interface. $\mathcal{K}$ has three interfaces $A$, $B$ and $E$.
It outputs a perfectly uniform secret key $k$ at the interfaces $A$ and $B$ and
 it has no functionality at the adversary's interface $E$.
The QAS protocol which includes the pair of converters $(\pi^{\textnormal{e}}_{A},\pi^{\textnormal{d}}_{B})$ constructs $\mathcal{C}^{\textnormal{q-sec}}$ given $\mathcal{C}^{\textnormal{q-insec}}$ and $\mathcal{K}$ within error $\epsilon^{\textnormal{q-sec}},$ if the following equation is satisfied:
 \begin{equation}
 \pi^{\textnormal{e}}_{A}\pi^{\textnormal{d}}_{B}(\mathcal{K}||\mathcal{C}^{\textnormal{q-insec}}) \underset{\epsilon^{\textnormal{q-sec}}}{\approx}\mathcal{C}^{\textnormal{q-sec}}\sigma^ {\textnormal{q-sec}}_E, \label{sim-auth}
 \end{equation}
where $\pi^{\textnormal{e}}_{A}$ is a converter which receives input state $\rho$ at its outside interface. It requests a key $k$ from $\mathcal{K}$, applies a map $E_k$ on $\rho$ which outputs an authenticated state $\tau$ and sends it through $\mathcal{C}^{\textnormal{q-insec}}.$
$\pi^{\textnormal{d}}_B$ takes in a (possibly) tampered state $\tau'$ and applies a verification map $D_k$ and outputs the state $\rho'.$  $\rho'$ is equal to $\rho$ if the verification succeeds and is equal to $\ket{\text{err}}\bra{\text{err}}$ otherwise (Figure~\ref{fig:auth-ideal} (b)).
The converter $\sigma^{\textnormal{q-sec}}_E$ (shown in Figure~\ref{fig:auth-ideal} (a)) is a simulator that given access only to the $E$ interface of $\mathcal{C}^{\textnormal{q-sec}}$ is able to emulate the
behaviour of the adversary so that the real-world protocol (involving the adversary) is statistically
indistinguishable from the ideal-world protocol (involving the simulator).

\begin{figure}[!h]
\centering  \includegraphics[width=\columnwidth]{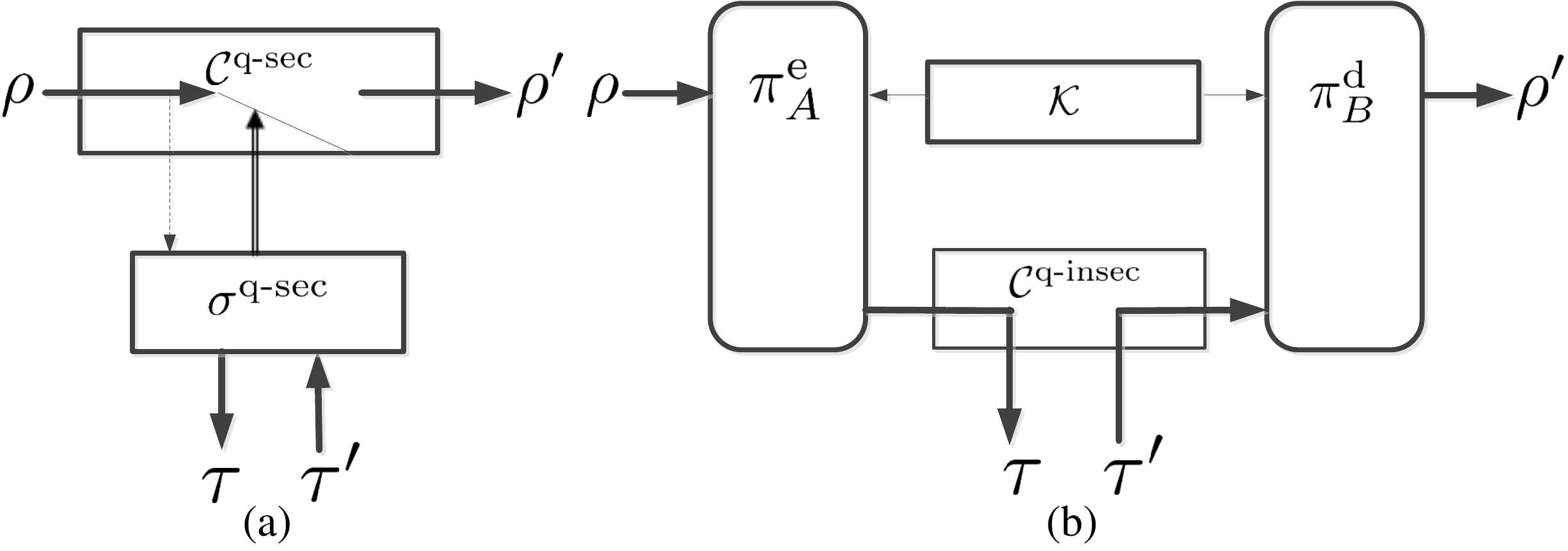}
\caption{ a). A schematic of the ideal resource (top) and simulator (bottom) for a secure quantum channel. Alice, Bob and Eve have access to the left, right and bottom interfaces of ideal resource $\mathcal{C}^{\textnormal{q-sec}},$ respectively. A simulator $\sigma^{\textnormal{q-sec}}_E$ can be appended to Eve's interface so that the interfaces of the compound object matches those of concrete protocol in (b). b) The QAS protocol $(\pi^{\textnormal{e}}_{A},\pi^{\textnormal{d}}_{B})$ constructs $\mathcal{C}^{\textnormal{q-sec}}$ given $\mathcal{C}^{\textnormal{q-insec}}$ and $\mathcal{K}.$}
\label{fig:auth-ideal}
\end{figure}

\section{Multi-Client Delegated Quantum Computation}
\label{sec:MCDQC}
In this section, we present definitions for multi-client DQC.
 In this scheme, $n$ clients, $\text{Client}_1, \text{Client}_2,...,$ and $\text{Client}_n$ delegate a global computation to an untrusted powerful server in a blind verifiable manner. $\text{Client}_i$ is responsible for providing their own part of the input state $\rho_i$.
The global computation is comprised of a network of local unitaries, individually defined by each client.
The clients' computations may interlock in the sense that their local unitaries may affect qubits that are (part of) the outputs of other clients' local unitaries. We call these qubits, \textit{common} qubits. The output of the global computation is split into $n$ components such that every $\text{Client}_i$ receives one part $\varrho_i.$ In addition to security concerns on the server side, the privacy among the clients is required. Each client learns no information about the input, local unitaries and output states of the other clients apart from what can be learned by their input, their local unitary outputs, and the final output state.

The multiple clients may interact in many configurations which are proven to be universal, such as nearest neighbour topology, where each client interacts only with their nearest neighbours. However, a desired topology for a multi-client DQC should be i) a regular arrangement which can hide the structure of the protocol such that all that is revealed is the total size of unitaries defined by clients in each round and ii) no dishonest client can become a communication bottleneck, such that the client can block another client from communicating with the others. We consider a fully connected topology~\cite{tanenbaum2003computer} which fulfils these requirements. A fully connected network is a communication network in which each client is connected to all other ones and is represented by a complete graph in graph theory. A fully connected network is symmetric, in the sense that there is no distinguished node, i.e., every party's view of the rest of the network is the same.

Let us define the set of $m$ local unitaries for $\text{Client}_i$ as $\mathcal{U}_{i}=\{U_{i}^{(h)}\}_{h\in\{1,\dots,m\}},$ where $U_{i}^{(h)}$ is the $h^{\text{th}}$ local unitary associated with $C_i$. If the number of local unitaries of the clients is not the same, then identities are appended to make them equal.
Figure~\ref{fig:MultiParty} shows an example of such a setting, with $n=3$ and $m=2.$ The set of labels of qubits of $U_{i}^{(h)}$ that are fed into $U_{j}^{(h+1)}$ is denoted by $T_{i \rightarrow j}^{(h)}$.
Now we are ready to define the ideal multi-client DQC resource. In the definition, $\psi_{i}, i\in\{1,\dots,n\}$ describes the quantum input $\rho_{i}$, the set of local unitaries $\mathcal{U}_{i}$ and the sets $T_{i \rightarrow j}^{(h)},$ where $j\in\{1,\dots,n\}$, and $h\in\{1,\dots,m-1\}.$

 \begin{figure}[!h]
 \centering
  \includegraphics[width=3 in]{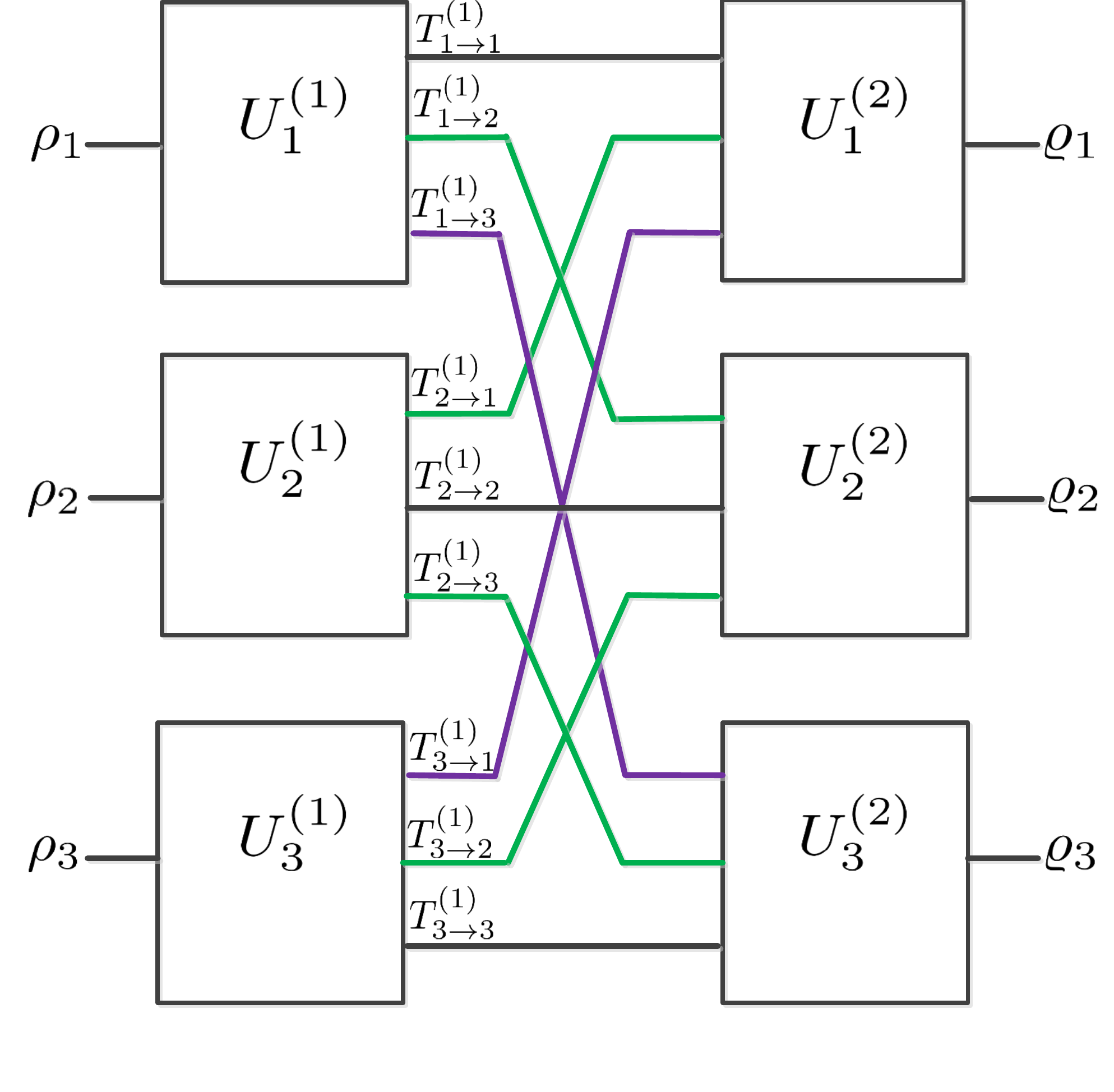}
   \caption{An example of multi-party computation framework with $n=3$ and $m=2$}
  \label{fig:MultiParty}
\end{figure}

 \begin{definition}
 The ideal DQC resource with $n$ clients, $Client_1, Client_2,\dots, Client_n$, which provides correctness, blindness, verifiability and privacy among clients is denoted by $\mathcal{S}^{n\textnormal{-bv}}$. $\mathcal{S}^{n\textnormal{-bv}}$ has $n+1$ interfaces, $C_1,\dots,C_n$ and $S$ standing for $Client_1,\dots,Client_n$ and the server respectively.
It takes $\psi_{i}$ at $C_i$'s interface.
The server's interface has a filter controlled bit $f^{n\textnormal{-bv}}_s$. If the server sets $f^{n\textnormal{-bv}}_{s}=0,$ $\mathcal{S}^{n\textnormal{-bv}}$ produces the correct output result of each client at its interface. If $f^{n\textnormal{-bv}}_s=1,$ $\mathcal{S}^{n\textnormal{-bv}}$ outputs the allowed leakage $\ell^{n\textnormal{-bv}}$ at the server's interface and some fixed error state $\ket{\emph{err}}\bra{\emph{err}}$ that is orthogonal to the output space, at every client's interface ($C_1,\dots,$ and $C_n$) for which an output is expected.
 \end{definition}

 After defining the ideal resource that a generic multi-client DQC protocol is expected to construct, we mention the resources that are required to achieve this, namely, a secret key resource $\mathcal{K},$ a two-way (quantum and classical) channel $\mathcal{R}_{i}$ between the server and each
 $\text{Client}_i$, and an authenticated classical channel $\mathcal{C}^{\textnormal{c-auth}}$ between any pair of clients. $\mathcal{C}^{\textnormal{c-auth}}$ is a channel which allows a classical message to be transmitted from the sender to the reciever. We denote all of the classical authentication channels among clients as $\mathcal{C}^{\textnormal{c-auth}}_{\textnormal{a}}.$

A protocol $\pi^{n\textnormal{-bv}}=(\pi_{c_1}^{n\textnormal{-bv}},\dots,\pi_{c_n}^{n\textnormal{-bv}},\pi_s^{n\textnormal{-bv}})$ constructs a blind verifiable resource $ \mathcal{S}^{n\textnormal{-bv}}$ from $(\mathcal{K}||\mathcal{C}^{\textnormal{c-auth}}_{\textnormal{a}}||\mathcal{R}_{1}|| \dots||\mathcal{R}_{n})$ within $\epsilon^{n\textnormal{-bv}},$ if the following conditions are satisfied:
\begin{itemize}
\item $\pi_{c_1}^{n\textnormal{-bv}}\cdots \pi_{c_n}^{n\textnormal{-bv}}\pi_s^{n\textnormal{-bv}}(\mathcal{K}||\mathcal{C}^{\textnormal{c-auth}}_{\textnormal{a}}||\mathcal{R}_{1}|| \dots||\mathcal{R}_{n})\underset{\epsilon^{n\textnormal{-bv}}}{\approx} \mathcal{S}^{n\textnormal{-bv}}\perp_{s},$
\item $
\pi_{c_1}^{n\textnormal{-bv}} \dots\pi_{c_n}^{n\textnormal{-bv}}(\mathcal{K}||\mathcal{C}^{\textnormal{c-auth}}_\textnormal{a}||\mathcal{R}_{1}|| \dots||\mathcal{R}_{n})\underset{\epsilon^{n\textnormal{-bv}}}{\approx}  \mathcal{S}^{n\textnormal{-bv}}\sigma^{n\textnormal{-bv}}_s.$
\end{itemize}
This ensures security against a dishonest server, and is what we shall focus on proving. A stronger security definition which also captures cheating clients is obtained by replacing the second criterion above with 
\begin{itemize}
	\item $\pi_{c_{h_1}}^{n\textnormal{-bv}} \cdots \pi_{c_{h_H}}^{n\textnormal{-bv}}(\mathcal{K}||\mathcal{C}^{\textnormal{c-auth}}_\textnormal{a}||\mathcal{R}_{1}|| \dots||\mathcal{R}_{n})\underset{\epsilon^{n\textnormal{-bv}}}{\approx} \mathcal{S}^{n\textnormal{-bv}}\sigma^{n\textnormal{-bv}}_{D}$,
\end{itemize}
 where the simulator $\sigma^{n\textnormal{-bv}}_{D}$ acts on the interfaces of all parties except $C_{h_1}$, $\dots$, $C_{h_{H-1}}$ and $C_{h_H}$. A full exploration of protocols for this later model of security is beyond the scope of the current work.

 We build the concrete protocol out of subprotocols that already exist. More specifically, the subprotocols consist of any BV-DQC and QAS, without making reference to any particular protocol. Then we use the existing results on the composable security of BV-DQC and QAS to argue that the composition we describe is secure. In terms of these subprotocols, the outline of the concrete setting is as follows.
 Each client delegates the computation of each local unitary to the server using BV-DQC. If any client detects a verification error, they broadcast $e=1$ via $\mathcal{C}^{\textnormal{c-auth}}_{\textnormal{a}}$ to all other clients, they all output $\ket{\text{err}}\bra{\text{err}}$ and abort the protocol. Otherwise, the client sends the common qubits to the server, and the server sends each qubit to the corresponding clients. The mentioned channel between the server and the client is insecure because the server might tamper the qubits and therefore, the client must protect the qubits with quantum authentication codes.

  At the beginning of the protocol, each client agrees on a key for each set of common qubits with the desired recipient. For the first layer of unitaries, true computation followed by an authentication map is delegated to the server.
  In the intermediate layers, the qubits are first decoded, the true computation is applied, and finally the qubits are encoded again. As for the last layer of local unitary, the qubits are first decoded and then the true computation is applied. All of the aforementioned encoding and decoding maps are carried out by the agreed corresponding keys and as a part of computation. Similar to the case of verification error, if any client detects an authentication error between rounds, they broadcast $e=1$ via $\mathcal{C}^{\textnormal{c-auth}}_{\textnormal{a}}$ to other clients, they output $\ket{\text{err}}\bra{\text{err}}$ and abort the protocol.

 A generic protocol for secure multi-client DQC is described in Protocol 1. We note that in order to be secure against dishonest clients, the protocol would need to be ammended so that any client receiving e=1 should immediately rebroadcast this to all other clients before aborting the protocol to avoid a situation where a malicious client could cause only a subset of honest clients to abort.

\begin{algorithm} [htbp]
\caption{Blind Verifiable Multi-Client DQC}
\begin{algorithmic}[1]
\STATE Let $n$ and $m$ be the number of clients and the total number of local unitaries of each client respectively.
\STATE Let $ i,j\in\{1,...,n\}, h\in\{1,...,m\}, \text{and} \;h'\in\{1,...,m-1\}.$
\STATE Let $\rho_{i}$, $\mathcal{U}_{i}=\{U_{i}^{(h)}\},$ and $\varrho_{i}^{(m)}$ be the input, the set of local unitaries and 
the output of $\text{Client}_{i}$, respectively.
\STATE Let $T_{i \rightarrow j}^{(h')}$ be the set of labels of output qubits of $U_{i}^{(h')}$ that are fed into $U_{j}^{(h'+1)}.$
\STATE Let $(E_{k},D_{k})$ be a pair of keyed operations to be used to authenticate quantum states.

\STATE $\text{Client}_{i}$ and $\text{Client}_{j}$ ($i \ne j$) agree on a secret key $k_{i\rightarrow j}^{(h')}$ for authenticating the qubits with labels in $T_{i\rightarrow j}^{(h')}.$

\STATE Let \[
 V_{i}^{(h)}  = \left\{ \begin{array}{l}
\mathop {\mathop  \otimes \limits_{j}} E_{k_{i \rightarrow j}^{(h)}}U_{i}^{(h)} ,\,\,if\,h=1, \\
 \mathop {\mathop  \otimes \limits_{j} } E_{k_{i \rightarrow j}^{(h)}}U_{i}^{(h)}\mathop {\mathop  \otimes \limits_{j} } D_{k_{j \rightarrow i}^{(h-1)}},\,\,if\,1<h<m,\\
U_{i}^{(h)}\mathop {\mathop  \otimes \limits_{j} } D_{k_{j \rightarrow i}^{(h-1)}},\,if\,h=m, \\
 \end{array} \right.
 \]
where $E_{k_{i \rightarrow j}^{(h)}}$ and $D_{k_{i \rightarrow j}^{(h)}}$ are applied to the qubits with labels in $T_{i\rightarrow j}^{(h)}.$
\FOR {$h':= 1$ to $m-1$}
\FOR {$i:= 1$ to $n$}
\STATE $\text{Client}_{i}$ and the server run BV-DQC($\rho_{i}^{(h')},V_{i}^{(h')}$) which outputs $\varrho_{i}^{(h')}.$
\IF {$\varrho_{i}^{(h')}$ is $\ket{\text{err}}\bra{\text{err}}$}
\STATE $\text{Client}_{i}$ sends $e=1$ to the other clients through an authenticated classical channel and all clients output $\ket{\text{err}}\bra{\text{err}}$ and abort the protocol
\ELSE
\STATE $\text{Client}_{i}$ sends $\varrho_{i}^{(h')}$ excluding the qubits with labels in $T_{i \rightarrow i}^{(h')}$ to the server
\ENDIF
\FOR {$j:=1$ to $n$, $j\ne i$}
\STATE  Let $\varrho_{i \rightarrow j}^{(h')}$ be the state of the qubits of $\varrho_{i}^{(h')}$ in $T_{i \rightarrow j}^{(h')}$
\STATE {The server sends $\varrho_{i \rightarrow j}^{(h')}$ to $\text{Client}_{j}$}
\ENDFOR

\ENDFOR
\FOR {$i:= 1$ to $n$}
\STATE Let $\rho_{i}^{(h'+1)}=\mathop {\mathop  \otimes \limits_{j} } \varrho_{j \rightarrow i}^{(h')}.$
\ENDFOR
\ENDFOR
\FOR {$i:= 1$ to $n$}
\STATE $\text{Client}_{i}$ and the server run BV-DQC($\rho_{i}^{(m)},V_{i}^{(m)}$) which outputs $\varrho_{i}^{(m)}$
\IF {$\varrho_{i}^{(m)}$ is $\ket{\text{err}}\bra{\text{err}}$}
\STATE $\text{Client}_{i}$ sends $e=1$ to the other clients through an authenticated classical channel and all clients output $\ket{\text{err}}\bra{\text{err}}$ and abort the protocol
\ELSE
\STATE $\text{Client}_{i}$ outputs $\varrho_{i}^{(m)}$
\ENDIF
\ENDFOR
\end{algorithmic}
\end{algorithm}
\subsection{Security Proof}
\begin{theorem}
Protocol 1 constructs $\mathcal{S}^{n\textnormal{-bv}}$ using resources $\mathcal{K},$ $\mathcal{R}_1,$ $\mathcal{R}_2$ and $\mathcal{C}^{\textnormal{c-auth}}_{\textnormal{a}},$
 with error $\epsilon^{\textnormal{n-bv}}=(mn)\epsilon^{\textnormal{bv}}+n(n-1)(m-1)\epsilon^{\textnormal{q-sec}},$ where $n$ and $m$ are the numbers of clients and each client's unitaries respectively.
\end{theorem}
\begin{proof}
We will prove the security for a simple example with two clients, $\text{Client}_1$ and $\text{Client}_2$ each with one local unitary, $U_{c_1}$ and $U_{c_2}$ as shown in Figure~\ref{fig:SampleMBQC}. The proof can be directly extended for the general case with multiple clients.
The definition of the ideal resource related to the example in Figure~\ref{fig:whole-conc} is as follows. It should be noted that for simplicity we will write $U\rho$ instead of $U\rho U^{\dagger}$ each time we talk about applying a unitary operation $U$ to a quantum state $\rho$.

\begin{figure}[!h]
\centering  \includegraphics[width=3 in]{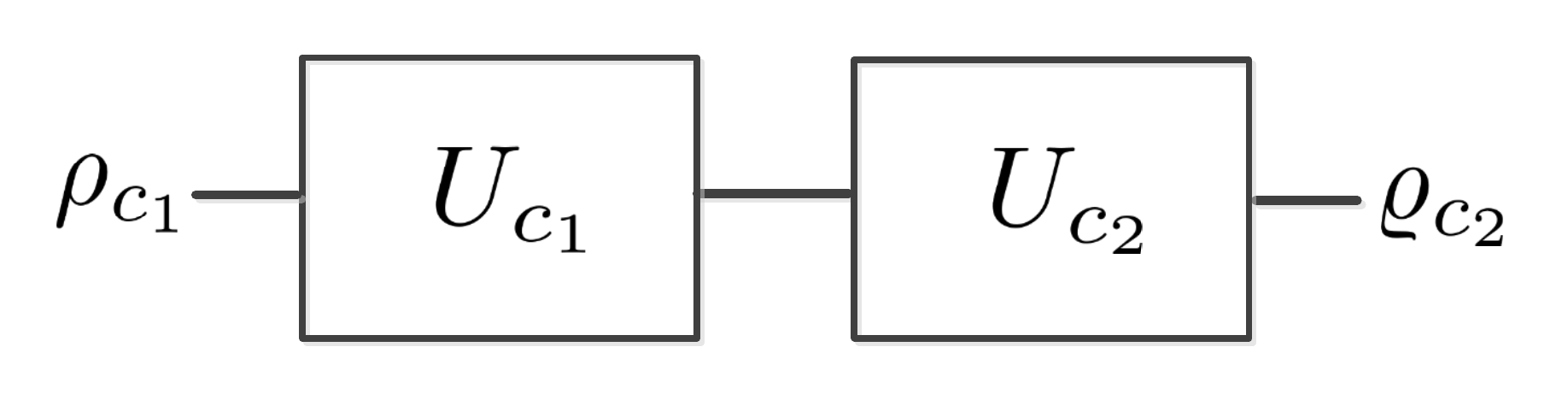}
\caption{A sample setting for multi-party computation with two clients, each with one local unitary.}
 \label{fig:SampleMBQC}
\end{figure}
\begin{definition}
The ideal 2-client DQC resource related to the example in Figure~\ref{fig:SampleMBQC}, $\mathcal{S}^{\textnormal{2-bv}},$ which provides correctness, blindness, verifiability and privacy among clients has 3 interfaces, $C_1$, $C_2$ and $S$ standing for ${Client}_1$, ${Client}_2$ and the server.
It takes inputs $\rho_{c_1}$ and $U_{c_1}$ at the $C_1$ interface, and $U_{c_2}$ at the $C_2$ interface. The server's interface has a filter controlled bit $f^{\textnormal{2-bv}}_s$.
If the server sets $f^{\textnormal{2-bv}}_s=0,$ $\mathcal{S}^{\textnormal{2-bv}}$ outputs
$\varrho_{c_2}=U_{c_2}U_{c_1}\rho_{c_1}$ at the $C_2$ interface.
If $f^{\textnormal{2-bv}}_s=1$, the server's interface outputs the allowed leakage and some fixed error state $\left| \emph{err} \right\rangle \left\langle\emph{err} \right|$ at the $C_2$ interface.
\end{definition}

The top of Figure~\ref{fig:ideal-sim} shows the ideal resource $\mathcal{S}^{2\textnormal{-bv}}$. We show that Protocol 2 constructs $\mathcal{S}^{2\textnormal{-bv}}$ using a parallel composition of $\mathcal{K},$ $\mathcal{R}_1,$ $\mathcal{R}_2$ and $\mathcal{C}^{\textnormal{c-auth}},$
 with error $\epsilon^{2\textnormal{-bv}}=\epsilon^{\textnormal{q-sec}}+2\epsilon^{\textnormal{bv}}.$  The proof of security proceeds in three main steps. In the first step, it is supposed that the parties have access to quantum capabilities, and hence they do not need the server. In the next step, we move towards the real setting by assuming that the parties do not have quantum capabilities, but delegate the computations to an ideal resource without any error. Finally in the last step, we reach to the real case where the ideal resource may produce error instead of the correct output.
       \begin{figure}[!h]
\centering  \includegraphics[width=4 in]{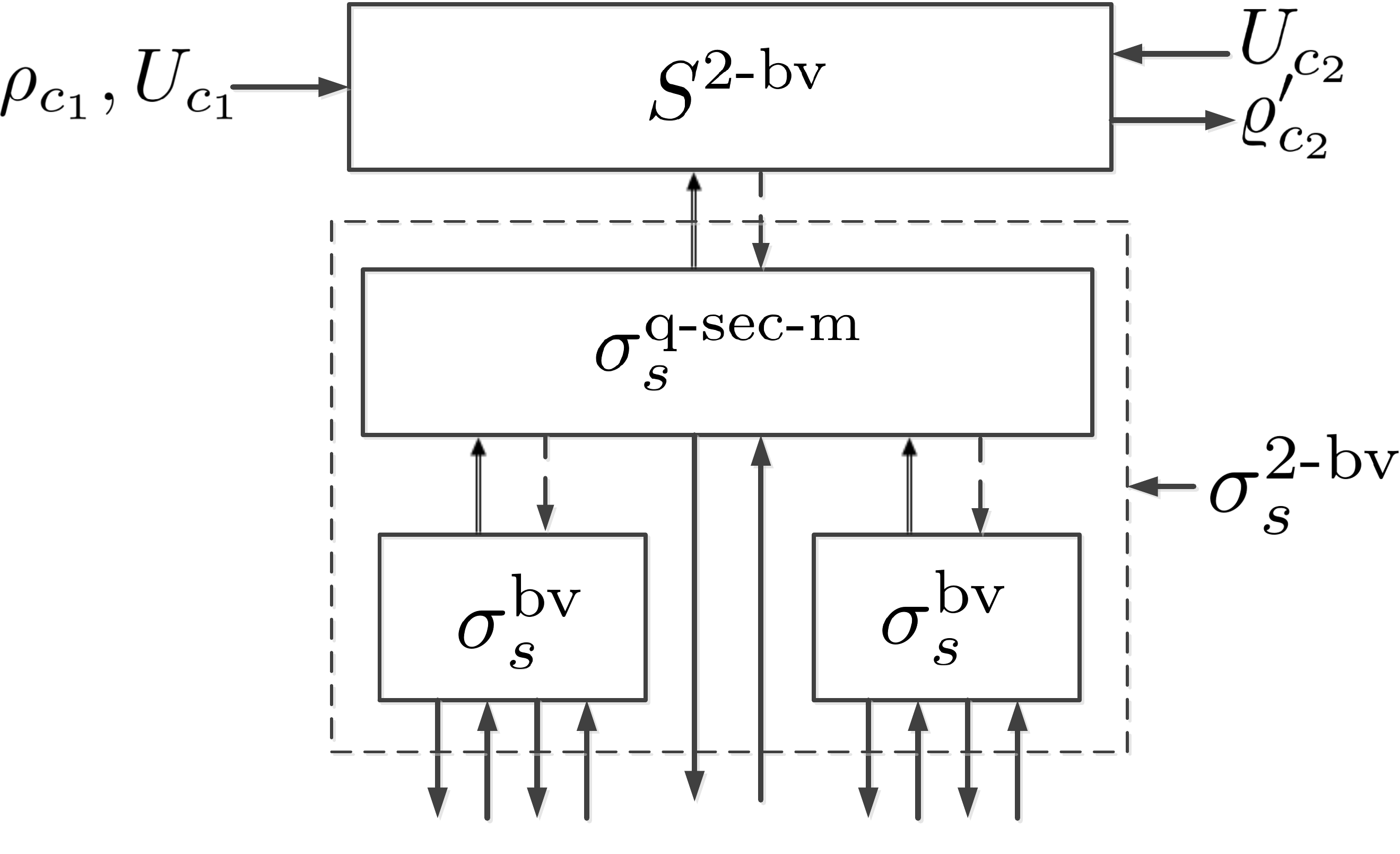}
\caption{The top of the figure and the region within the dashed box indicate ideal resource $\mathcal{S}^{2\textnormal{-bv}}$ and the simulator $\sigma^{2\textnormal{-bv}}_s$, respectively.}
 \label{fig:ideal-sim}
\end{figure}

\begin{algorithm}[!h]
\caption{Blind Verifiable Two-Client DQC for the Setting in Figure~\ref{fig:SampleMBQC}}
\begin{algorithmic}[1]
\STATE Let $\rho_{c_i}$ and $U_{c_i}$ be the input and local unitary of $C_i$, respectively, where $i\in\{1,2\}.$

 \STATE $\text{Client}_1$ and $\text{Client}_2$ agree on a secret key $k.$
\STATE Let $(E_{k},D_{k})$ be a pair of operations for authentication.

\STATE $\text{Client}_1$ and the server run BV-DQC($\rho_{c_1}, E_kU_{c_1}$) which outputs $\varrho_{1}^{(1)}.$
\IF {$\varrho_{c_1}$ is $\ket{\text{err}}\bra{\text{err}}$}
\STATE $\text{Client}_1$ sends $e=1$ to the other client through an authenticated classical channel and $\text{Client}_2$ outputs $\ket{\text{err}}\bra{\text{err}}$ and aborts the protocol
\ELSE
\STATE $\text{Client}_1$ sends $\varrho_1^{(1)}$ to the server
\ENDIF
\STATE The server sends $\varrho_1^{(1)}$ to $\text{Client}_2$
\STATE $\text{Client}_2$ and the server run BV-DQC($\varrho_1^{(1)},U_{c_2}D_k$) which outputs $\varrho'_{{c_2}}$
\STATE $\text{Client}_2$ outputs $\varrho'_{c_2}$

\end{algorithmic}
\end{algorithm}

  \begin{figure}[!h]
\centering  \includegraphics[width=\columnwidth]{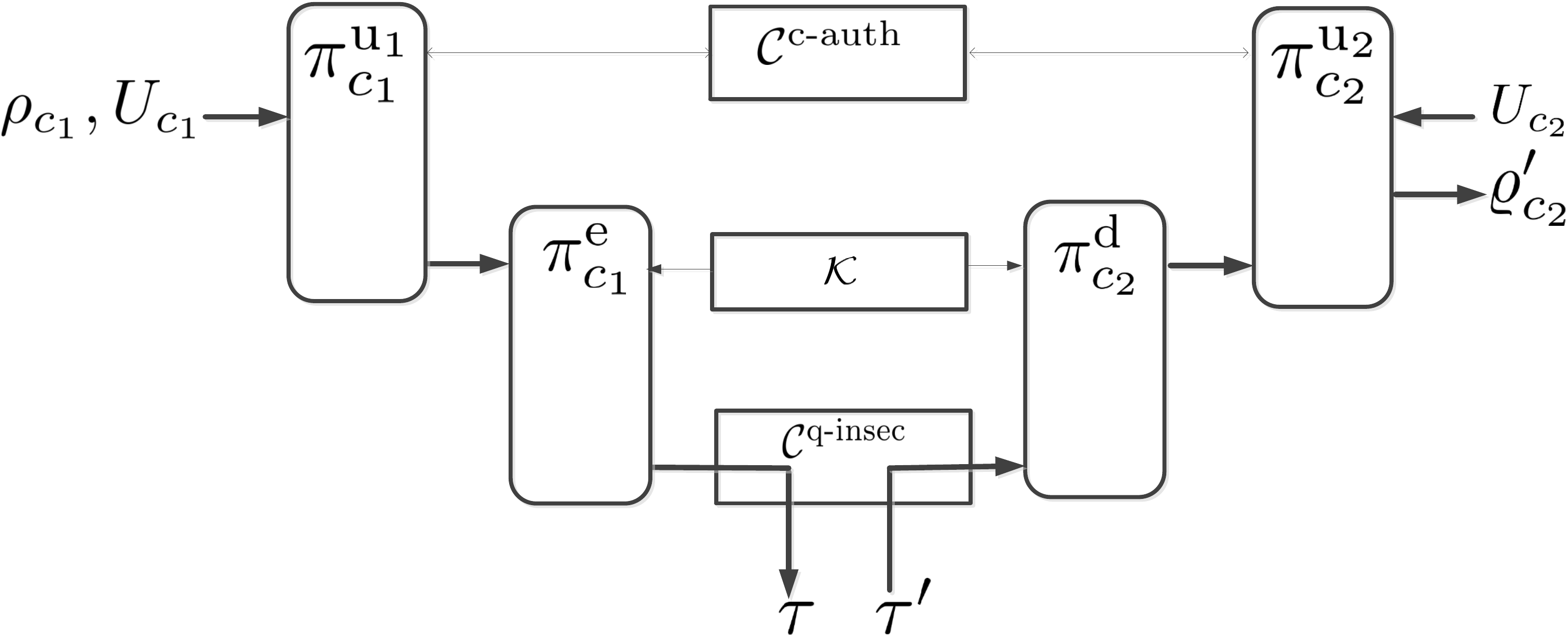}
\caption{The concrete protocol for implementing the sample example in Figure~\ref{fig:SampleMBQC}, where the parties are supposed to have access to quantum capabilities}
 \label{fig:conc-q-power}
\end{figure}

 1) First we suppose that the parties have access to quantum capabilities and describe the concrete protocol (shown in Figure~\ref{fig:conc-q-power}).
  In this figure, $\pi_{c_1}^{\textnormal{u}_1}$ is a converter which receives the input state $\rho_{c_1}$, the unitary $U_{c_1}$, and produces
 $U_{c_1}\rho_{c_1}.$ The converter $\pi^\textnormal{e}_{c_1}$ receives $U_{c_1}\rho_{c_1}$ at its inside interface. It requests a key $k$ from $\mathcal{K}$. It applies $E_k$ to $U_{c_1}\rho_{c_1}$ which outputs an authenticated state $\tau$ and sends it through $\mathcal{C}^{\textnormal{q-insec}}.$
 The converter $ \pi^{\textnormal{d}}_{c_2},$ takes a (possibly) tampered state $\tau'$ on the insecure channel and applies verification $D_k$. It outputs a state $\rho'$, which is equal to  $U_{c_1}\rho_{c_1}$ or $\ket{\text{err}}\bra{\text{err}}$
  depending on whether the verification succeeded or failed, respectively.
The converter $\pi_{c_2}^{\textnormal{u}_2}$ receives $\rho'$ at its inside interface and $U_{c_2}$ at its outside interface and outputs $\varrho'_{c_2}$. If $\rho'$ is $\ket{\text{err}}\bra{\text{err}},$ $\varrho'_{c_2}$ is equal to $\ket{\text{err}}\bra{\text{err}}$, otherwise it is $U_{c_2}U_{c_1}\rho_{c_1}.$ Using Eq. (\ref{sim-auth}), we have:
\begin{equation}
\pi_{c_1}^{\textnormal{u}_1}\pi_{c_2}^{\textnormal{u}_2} \pi^{\textnormal{e}}_{c_1}\pi^{\textnormal{d}}_{c_2}(\mathcal{K}||\mathcal{C}^{\textnormal{q-insec}})\underset{\epsilon^{\textnormal{q-sec}}}{\approx} \pi_{c_1}^{\textnormal{u}_1}\pi_{c_2}^{\textnormal{u}_2} \mathcal{C}^{\textnormal{q-sec}}\sigma^{\textnormal{q-sec}}_s.
\label{eq:converter1}
\end{equation}
Eq.~(\ref{eq:converter1}) holds because based on Eq.~(\ref{eq:non-inc-conv}), adding converters $\pi_{c_1}^{\textnormal{u}_1}$ and $\pi_{c_2}^{\textnormal{u}_2}$ can only decrease the distance between the two systems.
 Also we have:

\begin{equation}
\pi_{c_1}^{\textnormal{u}_1}\pi_{c_2}^{\textnormal{u}_2} \mathcal{C}^{\textnormal{q-sec}}\sigma^{\textnormal{q-sec}}_s= \mathcal{S}'\sigma^{\textnormal{q-sec}}_s,
 \label{eq:ideal-interghange1}
 \end{equation}
where the ideal resource, $\mathcal{S}'$ is identical to $\mathcal{C}^{\textnormal{q-sec}},$ except that $\pi_{c_1}^{\textnormal{u}_1}$ and $\pi_{c_2}^{\textnormal{u}_2}$ have been merged with it.

 2) In Figure~\ref{fig:conc-q-power}, the quantum computations performed by the parties are modelled as converters. We move forward to the real case where the parties do not have access to quantum capabilities. If they want to delegate this to a server it should be modelled as a resource. Now we suppose the clients have access to two resources, $\mathcal{S}^ \mathcal{Q}_1$ and $\mathcal{S}^ \mathcal{Q}_2$ which perform the quantum part of the computation, without any error. So we have:
\begin{equation}
\pi^{\textnormal{e}'}_{c_1}\pi^{\textnormal{d}'}_{c_2}(\mathcal{K}||\mathcal{C}^{\textnormal{q-insec}}||\mathcal{S}^ \mathcal{Q}_1|| \mathcal{S}^\mathcal{Q}_2)=\pi_{c_1}^{\textnormal{u}_1}\pi_{c_2}^{\textnormal{u}_2} \pi^{\textnormal{e}}_{c_1} \pi^{\textnormal{d}}_{c_2}(\mathcal{K}||\mathcal{C}^{\textnormal{q-insec}}),
\label{eq:SQ1}
\end{equation}
where $\pi^{\textnormal{e}'}_{c_1}$ and $\pi^{\textnormal{d}'}_{c_2}$ are now the classical (or almost classical) converters of $\text{Client}_1$ and $\text{Client}_2$ that capture their local operations:
 $\pi^{\textnormal{e}'}_{c_1}$ receives an input state $\rho_{c_1}$ and $U_{c_1}$ at its outer interface, requests a key $k$  and sends $\rho_{c_1}$ and  $E_kU_{c_1}$ to $\mathcal{S}_1^\mathcal{Q}$. It receives the response $\tau$ and sends it to $\pi^{\textnormal{d}'}_{c_2}$ via an insecure channel $\mathcal{C}^{\textnormal{q-insec}}$. $\pi^{\textnormal{d}'}_{c_2}$ receives $\tau'$ on the insecure channel and $U_{c_2}$ at its outer interface, requests the key $k$ from $\mathcal{K}$, sends $\tau'$ and $U_2D_k$ to $\mathcal{S}^\mathcal{Q}_2$, and receives a response $\varrho'_{c_2}$, which it outputs at its outer interface. $\varrho'_{c_2}$ is equal to $\ket{\text{err}}\bra{\text{err}}$ if the state has been tampered on the $\mathcal{C}^{\textnormal{q-insec}},$ otherwise it is equal to $U_{c_2}U_{c_1}\rho_{c_1}.$
 Combining Eqs.(\ref{eq:converter1}), (\ref{eq:ideal-interghange1}) and (\ref{eq:SQ1}) we get:
 \begin{equation}
\pi_{c_1}^{\textnormal{e}'} \pi_{c_2}^{\textnormal{d}'}(\mathcal{K}||\mathcal{C}^{\textnormal{q-insec}}||\mathcal{C}^{\textnormal{c-auth}}||\mathcal{S}^\mathcal{Q}_1 ||\mathcal{S}^\mathcal{Q}_2)\underset{\epsilon^{\textnormal{q-sec}}}{\approx} \mathcal{S}'\sigma^{\textnormal{q-sec}}_s.
 \label{eq1}
 \end{equation}
3) The problem we have now is that $\mathcal{S}^\mathcal{Q}_1$ and $\mathcal{S}^\mathcal{Q}_2$ are not the ideal resources $\mathcal{S}^{\textnormal{bv}}$ constructed by a BV-DQC protocol, because an error might occur during the computation. So $\mathcal{S}^\mathcal{Q}_1$ and $\mathcal{S}^\mathcal{Q}_1$ should be modified by giving the server the power to flip a switch and produce an error instead of the correct outcome. Now we have $\mathcal{S}^{\textnormal{bv}}$. Also, we need to modify $\pi^{e'}_{c_1}$ because now $\text{Client}_1$ could receive an error from $\mathcal{S}^{\textnormal{bv}}_1$. So the new converter $\pi^{e''}_{c_1}$ behaves identically to $\pi^{e'}_{c_1}$, except when it gets an error from $\mathcal{S}^\textnormal{bv}_1$, it sends $e=1$ to the other client.  Also $\pi^{d''}_{c_1}$ is the same as $\pi^{d'}_{c_1}$ but when it receives $e=1,$ it outputs $\ket{\text{err}}\bra{\text{err}}.$

 We also change the ideal resource $S'$ slightly. $S^{2\textnormal{-bv}}$ is identical to $S',$ the only difference is that it also leaks $\ell^{\textnormal{bv}}_1$ and $\ell^{\textnormal{bv}}_2$, and whenever there is an error in computation it outputs $\ket{\text{err}}\bra{\text{err}}$ at the $C_2$ interface.
We also have to modify the simulator $\sigma^{\textnormal{q-sec}}_s$ as follows:
$\sigma^{\textnormal{q-sec-m}}_s$ behaves the same as $\sigma^{\textnormal{q-sec}}_s,$ but it also outputs the leak of both computations $\ell^{\textnormal{bv}}_1$ and $\ell^{\textnormal{bv}}_2,$
and waits for the server to anounce if it does the correct computation or not. If so, $\sigma^{\textnormal{q-sec-m}}$ tells $\mathcal{S}^\textnormal{2-bv}$ to output the correct value, otherwise $\ket{\text{err}}\bra{\text{err}}$ at the $C_2$ interface.
We now need to prove that 
\begin{equation}
\pi_{c_1}^{\textnormal{e}''}\pi_{c_2}^{\textnormal{d}''}(\mathcal{K}||\mathcal{C}^{\textnormal{q-insec}}||\mathcal{C}^{\textnormal{c-auth}}||\mathcal{S}^\textnormal{bv}_1 ||\mathcal{S}^{\textnormal{bv}}_2)\underset{\epsilon^{\textnormal{q-sec}}}{\approx} \mathcal{S}^{\textnormal{2-bv}}\sigma^{\textnormal{q-sec-m}}_s.
\label{eq:S-bv1}
\end{equation}
Eq.~(\ref{eq:S-bv1}) follows from  Eq.~(\ref{eq1}), because an optimal strategy for the distinguisher consists of allowing the correct computation to be performed, which is the case covered by Eq.~(\ref{eq1}). If the distinguisher provokes an error, then $\text{Client}_2$ will simply output an error, and the distinguisher has less information available to try to distinguish the real from the ideal systems.  Also revealing the accept/reject bit does not change the distance
because the distinguisher has access to
all interfaces, which contain the
accept/reject information.
Based on Eq.~(\ref{eq:non-inc-conv}) plugging simulators $\sigma^{\textnormal{bv}}_s$ on the real and ideal systems cannot increase their distance, so we have
\begin{equation}
\pi_{c_1}^{\textnormal{e}''}\pi_{c_2}^{\textnormal{d}''}(\mathcal{K}||\mathcal{C}^{\textnormal{q-insec}}||\mathcal{C}^{\textnormal{c-auth}}||\mathcal{S}^{\textnormal{bv}}_1 ||\mathcal{S}^{\textnormal{bv}}_2)\sigma_{s}^{\textnormal{bv}}\sigma^{\textnormal{bv}}_{s}\underset{\epsilon^{\textnormal{q-sec}}}{\approx}\mathcal{S}^{\textnormal{2-bv}}\sigma^{\textnormal{q-sec-m}}\sigma_{s}^{\textnormal{bv}}\sigma_{s}^{\textnormal{bv}}.
\label{eq3}
\end{equation}
Now  we can move the brackets around in the left-hand side of Equation~(\ref{eq3}), which could equivalently have been written:
\begin{equation}
\pi^{e''}_{c_1}\pi^{d''}_{c_2}(\mathcal{K}||\mathcal{C}^{\textnormal{q-insec}}||\mathcal{C}^{\textnormal{c-auth}}||\mathcal{S}^\textnormal{bv}_1\sigma^{\textnormal{bv}}_s ||\mathcal{S}^{\textnormal{bv}}_2\sigma^{\textnormal{bv}}_s)\underset{\epsilon^{\textnormal{q-sec}}}{\approx}\mathcal{S}^{2\textnormal{-bv}}\sigma^{\textnormal{q-sec-m}}_s\sigma^{\textnormal{bv}}_s\sigma^{\textnormal{bv}}_s.
\label{1-Eq:8}
\end{equation}
The final step is to replace $\mathcal{S}^\textnormal{bv}_1\sigma^{\textnormal{bv}}_s$ and $\mathcal{S}^\textnormal{bv}_2\sigma^{\textnormal{bv}}_s$ by the real BV-DQC protocols.
It is shown in Figure~\ref{fig:whole-conc}. Combining Eqs,~(\ref{BV-DQC}) and~(\ref{1-Eq:8}) we obtain:
\[\pi^{e''}_{c_1}\pi^{d''}_{c_2}(\mathcal{K}||\mathcal{C}^{\textnormal{q-insec}}||\mathcal{C}^{\textnormal{c-auth}}||\pi_{c_{1}}^\textnormal{bv}\mathcal{R}_1||\pi_{c_{2}}^\textnormal{bv}\mathcal{R}_2)\underset{\epsilon^{\textnormal{q-sec}}+2\epsilon^{\textnormal{bv}}}{\approx} \mathcal{S}^{2\textnormal{-bv}}\sigma^{\textnormal{q-sec-m}}_s\sigma^{\textnormal{bv}}_s\sigma^{\textnormal{bv}}_s.\]
Moving brackets around we get:
 \[\pi_{c_1}^{e''}\pi_{c_2}^{d''}\pi_{c_{1}}^\textnormal{bv}\pi_{c_{2}}^\textnormal{bv}(\mathcal{K}||\mathcal{C}^{\textnormal{q-insec}}||\mathcal{C}^{\textnormal{c-auth}}||\mathcal{R}_1||\mathcal{R}_2)\underset{\epsilon^{\textnormal{q-sec}}+2\epsilon^{\textnormal{bv}}}{\approx} \mathcal{S}^\textnormal{2-bv}\sigma^{\textnormal{q-sec-m}}_s\sigma^{\textnormal{bv}}_s\sigma^{\textnormal{bv}}_s.\] Figure~\ref{fig:whole-conc} demonstrates the concrete protocol and resources.
Simply put, the composition of the authentication protocol and the BV-DQC protocols constructs the 2-client BV-DQC resource $\mathcal{S}^\textnormal{2-bv}$ with error equal to $\epsilon^{\textnormal{q-sec}}+2\epsilon^{\textnormal{bv}}$. The simulator $\sigma^\textnormal{2-bv}_s$ is equal to $\sigma^{\textnormal{q-sec-m}}_s\sigma^{\textnormal{bv}}_s\sigma^{\textnormal{bv}}_s.$

  \begin{figure}[!h]
\centering  \includegraphics[width=\columnwidth]{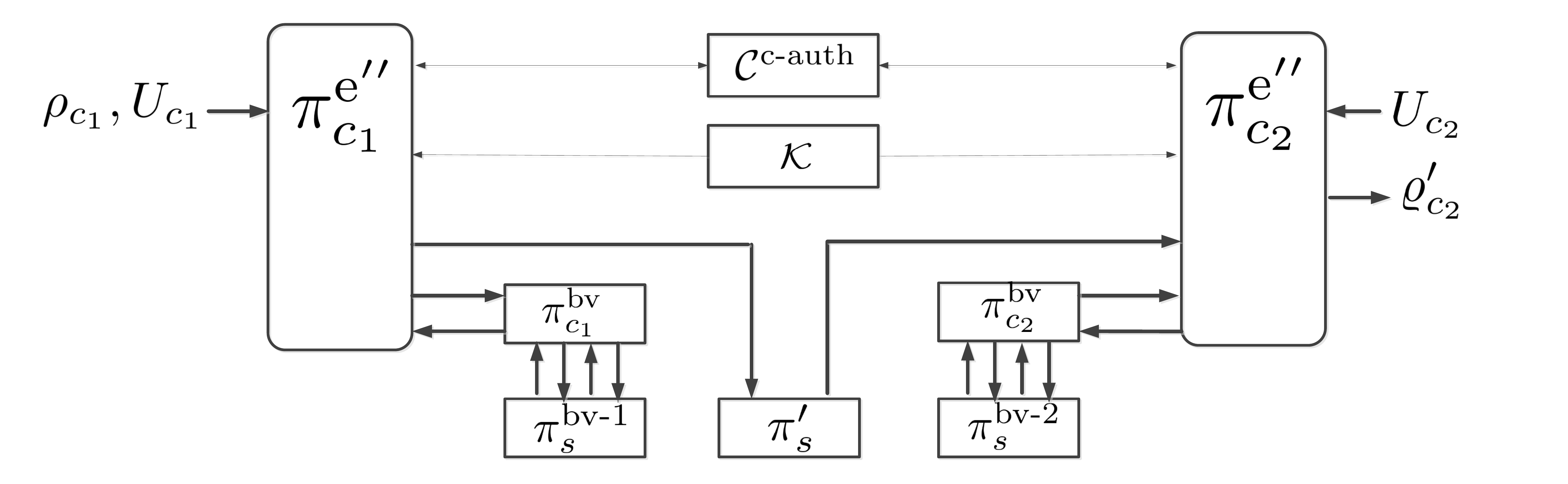}
\caption{The concrete protocol for implementing the example in Figure~\ref{fig:SampleMBQC}}
 \label{fig:whole-conc}
\end{figure}

The security proof can be directly extended to the case of a generic multi party DQC. In this case, the total error, $\epsilon^{\textnormal{n-bv}}$ is the sum of the errors of all runs of BV-DQC and QAS. In the case of fully connected network with $n$ clients and each with $m$ unitaries, it is equal to:
\[\epsilon^{\textnormal{n-bv}}=(mn)\epsilon^{\textnormal{bv}}+n(n-1)(m-1)\epsilon^{\textnormal{q-sec}},\]
where $mn$ is the number of unitaries performed by the server, and $n(n-1)(m-1)$ is the number of sets of non-common qubits which require quantum authentication codes to be protected.\;\;\;\; \;\;\;\; \;\;\;\; \;\;\;\;\;\;\;\; \;\;\;\; \;\;\;\; \;\;\;\;\;\;\;\; \;\;\;\; \;\;\;\; \;\;\;\;\;\;\;\; \;\;\;\; \;\;\;\;\;\;\;\;\;\;\;$\blacksquare$

\end{proof}

\section{Reducing the quantum communications}
\label{sec:RQC}
In this section, we make an additional assumption about BV-DQC scheme which results in reducing the quantum communication between the clients and the server. To achieve this goal, first we define the ideal resource $\mathcal{S}^{\textnormal{bb}}$ which is used as a building block in the desirable particular structure of the BV-DQC protocol.

   \begin{definition}
 The ideal DQC resource $\mathcal{S}^{\textnormal{bb}}$ has two interfaces $C$ and $S$ standing for the client and the server. It takes $U$ and $k'$ at the $C$ interface and $E_{k'}\rho$ at the $S$ interface, where $E_{k'}$ is an authentication map on qubits in $\rho$. The server also has a dishonest interface. It takes a filtered controlled bit $f^{\textnormal{bb}}$ at the $S$ interface.
 If server sets $f^{\textnormal{bb}}=0,$ $\mathcal{S}^{\textnormal{bb}}$ picks a uniformly random key $k$, outputs $k$ at the $C$ interface and applies
 $E_{k}UD_{k'}$ to
  $E_{k'}\rho$ and
  produces the honest result
 $E_{k}U\rho$ at the $S$ interfac. If $f^{\textnormal{bb}}=1,$ the resource outputs the allowed leakage $\ell^{\textnormal{bb}}$ at the $S$ interface and some fixed error state $\chi=\ket{\emph{err}}\bra{\emph{err}}$
  orthogonal to the output space at the $C$ interface.
   \end{definition}
Now suppose that the BV-DQC protocol $\pi^{\textnormal{bv}}$ is composed of the following converters as Figure~\ref{fig:bb} shows:
\begin{equation}
\pi^{\textnormal{bv}}=\pi^{\textnormal{e}}_c\pi_s^{\textnormal{bb}}\pi_c^{\textnormal{bb}}\pi^{\textnormal{d}}_c,
\end{equation}
 which means that client's part of the protocol, $\pi^{\textnormal{bv}}_c$, is decomposed into three parts; i) $\pi^{\textnormal{e}}_c,$ ii) $ \pi^{\textnormal{bb}}_ c,$ and iii) $\pi^{\textnormal{d}}_c$ as introduced in the following:

\begin{enumerate}[label=\roman*)]
  \item  $\pi^{\textnormal{e}}_c$ receives $\rho$ and $U$ at its outside interface. It requests a key $k'$ and sends $E_{k'}\rho$ to the server.
  \item $\pi^{\textnormal{bb}}_c$ receives $U$ and $k'$ from $\pi^{\textnormal{e}}_c$. $\pi_s^{\textnormal{bb}}$ receives $E_{k'}\rho$ on its inside interface. $\pi_s^{\textnormal{bb}}$ and $\pi_c^{\textnormal{bb}}$ communicate with each other through $\mathcal{R}$ to run the bulk of the DQC protocol, where $\pi^{\textnormal{e}}_c \pi^{\textnormal{bb}}_c\pi^{\textnormal{bb}}_s $ constructs $\pi^{\textnormal{e}}_c\mathcal{S}^{\textnormal{bb}}$ within $\epsilon^{\textnormal{bb}}$, i.e.,
      \begin{equation}
\pi^{\textnormal{e}}_c\pi_c^{\textnormal{bb}}\mathcal{R}\pi_s^{\textnormal{bb}}\underset{\epsilon^{\textnormal{bb}}}{\approx} \pi^{\textnormal{e}}_c \mathcal{S}^{\textnormal{bb}}\perp_s,\pi^{\textnormal{e}}_c \pi_c^{\textnormal{bb}}\mathcal{R}\underset{\epsilon^{\textnormal{bb}}}{\approx} \pi^{\textnormal{e}}_c\mathcal{S}^{\textnormal{bb}}\sigma^{\textnormal{bb}}_s,
 \label{BB-DQC}
\end{equation}
therefore we can replace $\pi_c^{\textnormal{bb}}\mathcal{R}\pi_s^{\textnormal{bb}}$ by $\mathcal{S}^{\textnormal{bb}}$ in any context where the client runs $\pi^{\textnormal{e}}_c$. In Eq. (\ref{BB-DQC}), we used the idea of context-restricted indifferentiability, which allows us to model a resource that does not compose generally but can only be used within a well-specified set of contexts \cite{cryptoeprint:2017:461}. $\pi_c^{\textnormal{bb}}$ passes two values to  $\pi^{\textnormal{d}}_c;$ the key $k$ and a bit which indicates whether any cheating has been detected so far.
  \item
  If the server behaves honestly, $\pi^{\textnormal{d}}_c$ receives $E_k U \rho$ from the server, it verifies the authenticity of this state, and if this test succeeds and
$\pi^{\textnormal{d}}_c$ did not detect any cheating either, it outputs the result of the decryption (which should be $U\rho$), otherwise it outputs an error.
\end{enumerate}
\begin{figure}[!h]
\centering  \includegraphics[width=4.5 in]{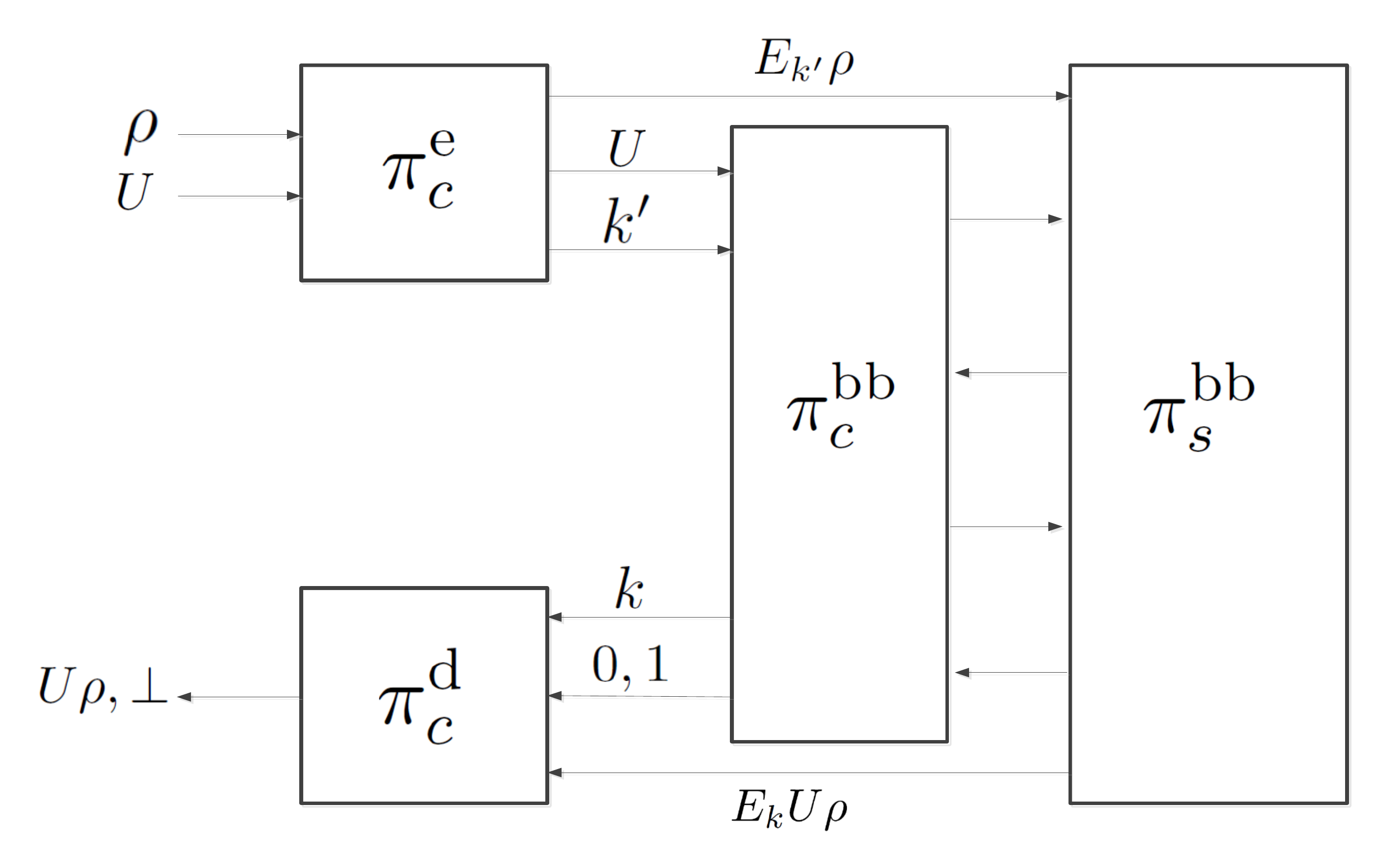}
\caption{ The converter related to the particular type of BV-DQC protocol which results in decreased overhead in the multi-client DQC protocol.}
\label{fig:bb}
\end{figure}
Protocol 3 is a modified blind verifiable multi-client DQC which removes the communications between the server and the clients between the rounds.
 Three types of protocols, BB1-DQC, BB-DQC, and BB2-DQC which are used in Protocol 3, are defined as follows:

\begin{itemize}
\renewcommand{\labelitemi}{$\bullet$}
\item
BB1-DQC($\rho,U$) means $\pi^{\textnormal{e}}_c\pi_s^{\textnormal{bb}}\pi_c^{\textnormal{bb}}$ is run which outputs $E_kU\rho$ and $k,$ if the verification passes.
\item
BB-DQC($E_{k'}\rho,U, k'$) means $\pi_s^{\textnormal{bb}}\pi_c^{\textnormal{bb}}$ is run which outputs $E_kU\rho$ and $k,$ if the verification passes.
\item
BB2-DQC($E_{k'}\rho,U, k'$) means $\pi_s^{\textnormal{bb}}\pi_c^{\textnormal{bb}}\pi_c^{\textnormal{d}}$ is run which outputs $U\rho,$ if the verification passes.
\end{itemize}

\begin{algorithm} [htbp]
\caption{Modified Blind Verifiable Multi-Client DQC (Removing the Communications Between the Clients and the Server)}
\begin{algorithmic}[1]
\STATE  Let $n$ and $m$ be the number of clients and the total number of local unitaries of each client, respectively.
\STATE Let $ i,j\in\{1,\dots,n\}, h\in\{1,\dots,m\}, \text{and}\;h'\in\{1,\dots,m-1\}.$
\STATE Let $\rho_{i}$, $\mathcal{U}_{i}=\{U_{i}^{(h)}\},$ and $\varrho_{i}$ be the input, the set of local unitaries and output of $\text{Client}_{i}$, respectively.
\STATE Let $T_{i \rightarrow j}^{(h')}$ be the set of labels of qubits of $U_{i}^{(h')}$ that are fed into $U_{j}^{(h'+1)}.$

\STATE Let $(E_{k},D_{k})$ be a pair of keyed operations using for authenticating quantum states.
\FOR {$i:=1$ to $n$}
\STATE $\text{Client}_{i}$ and the server run BB1-DQC($\rho_{i}, U_{i}^{(1)}$).
\IF{error is detected}
\STATE   $\text{Client}_{i}$ sends $e=1$ to the other clients through an authenticated classical channel and all clients abort the protocol.
\ELSE
\STATE the protocol outputs $\varrho_i^{(1)}= \mathop {\mathop  \otimes \limits_{j} } E_{k_{i \rightarrow j}^{(1)}} U_{i}^{(1)}\rho_{i}$ and $\{k^{(1)}_{i \rightarrow j}\}_{j\in\{1,...n\}}.$
\FOR {$j:=1$ to $n$}
\STATE $\text{Client}_{i}$ sends $k_{i \rightarrow j}^{(1)}$ to $\text{Client}_{j}.$
\STATE Let $\varrho_{i \rightarrow j}^{(1)}$ be the state of qubits in $T_{i \rightarrow j}^{(1)}.$
\ENDFOR
\ENDIF
\ENDFOR

\FOR {$i:= 1$ to $n$}
\STATE Let $\rho_{i}^{(2)}=\mathop {\mathop  \otimes \limits_{j} } \varrho_{j \rightarrow i}^{(1)}.$
\ENDFOR

\FOR {$h:= 2$ to $m-1$}
\FOR {$i:= 1$ to $n$}
\STATE $\text{Client}_{i}$ and the server run BB-DQC($\rho_{i}^{(h)},  U_i^{(h)} , \{k^{(h-1)}_{j \rightarrow i}\}_{j\in\{1,...n\}}$).

\IF {error is detected}
\STATE   $\text{Client}_{i}$ sends $e=1$ to the other clients through an authenticated classical channel and all clients abort the protocol.
\ELSE
\STATE the protocol outputs $\varrho_{i}^{(h)}=(\mathop {\mathop  \otimes \limits_{j} } E_{k_{i \rightarrow j}^{(h)}})U_{i}^{(h)}(\mathop {\mathop  \otimes \limits_{j} } D_{k_{j \rightarrow i}^{(h-1)}})\rho_{i}^{(h)}$ and $k_{i \rightarrow j}^{(h)}.$
\FOR {$j:=1$ to $n$}

\STATE $\text{Client}_{i}$ sends $k_{i \rightarrow j}^{(h)}$ to $\text{Client}_{j}.$
\STATE Let $\varrho_{i \rightarrow j}^{(h)}$ be the state of qubits in $T_{i \rightarrow j}^{(h)}.$
\ENDFOR
\ENDIF
\ENDFOR
\FOR {$i:= 1$ to $n$}
\STATE Let $\rho_{i}^{(h+1)}=\mathop {\mathop  \otimes \limits_{j} } \varrho_{j \rightarrow i}^{(h)}.$
\ENDFOR
\ENDFOR

\FOR {$i:= 1$ to $n$}
\STATE $\text{Client}_{i}$ and the server run BB2-DQC($\rho_{i}^{(m)},  U_i^{(m)} , \{k^{(m)}_{j \rightarrow i}\}_{j\in\{1,...n\}}$)
\IF {error is detected}
\STATE   $\text{Client}_{i}$ sends $e=1$ to the other clients through an authenticated classical channel and all clients abort the protocol.
\ELSE
\STATE   the protocol outputs $\varrho_{i}^{(m)}=U_i^{(m)}(\mathop {\mathop  \otimes \limits_{j} } D_{k_{j \rightarrow i}^{(m-1)}})\rho_i^{(m)} $
\ENDIF
\ENDFOR

\end{algorithmic}
\end{algorithm}

 \subsection{Security Proof}
\begin{theorem}
Protocol 3 constructs $\mathcal{S}^{n\textnormal{-bv}}$ using resources $\mathcal{K},$ $\mathcal{R}_1,$ $\mathcal{R}_2$ and $\mathcal{C}^{\textnormal{c-auth}}_{\textnormal{a}},$
 with error 
 $\epsilon^{\textnormal{n-bv}}=(mn)\epsilon^{\textnormal{bb}}+n(n-1)(m-1)\epsilon^{\textnormal{q-sec}}+2n\epsilon^{\textnormal{q-sec}},$ where $n$ and $m$ are the numbers of clients and each client's unitaries respectively.
\end{theorem}
\begin{proof}
 Again, we first prove the security for the example in Figure \ref{fig:SampleMBQC} with two clients. We show that Protocol 4 constructs $\mathcal{S}^{2\textnormal{-bv}}$ using $\mathcal{K},$ $\mathcal{R}_1,$ $\mathcal{R}_2$ and $\mathcal{C}^{\textnormal{c-auth}},$
 with error $\epsilon^{2\textnormal{-bv}}=2\epsilon^{\textnormal{bb}}+3\epsilon^{\textnormal{q-sec}}.$ The proof is easily extended to the general case. 
 The proof is essentially similar to the proof of the previous section. First, we suppose the parties have access to the quantum capabilities, but we decompose the protocol in a slightly different way of the previous section. The protocol is composed of the following converters:
\[\pi_{c_1}^{\textnormal{e}_1}\pi_{c_1}^{\textnormal{d}_1}\pi_{c_1}^{\textnormal{u}_1}\pi^{\textnormal{e}_2}_{c_1}\pi^{\textnormal{d}_2}_{c_2}\pi_{c_2}^{\textnormal{u}_2} \pi_{c_2}^{\textnormal{e}_3}\pi_{c_2}^{\textnormal{d}_3}.\]
The converter $\pi^{\textnormal{e}_1}_{c_1}$ receives $\rho_{c_1}$ at its outside interface. It requests a key $k_1$ from $\mathcal{K}$. It applies the map $E_{k_1}$ to $\rho_{c_1}$ which outputs the state $\tau_{1}.$
$\pi_{c_1}^{\textnormal{d}_1}$ receives $\tau'_{1}$ at its outside interface. It receives the key $k_1$ from $\mathcal{K}$ and applies $D_{k_1}$ which outputs $\rho_{c_1}$ or $\ket{\text{err}}\bra{\text{err}},$ depending on whether verification succeeds or fails. If verification succeeds, the converter $\pi_{c_1}^{\textnormal{u}_1}$ receives $U_{c_1}$ and $\rho_{c_1}$ at its outside and inside interfaces respectively and outputs $U_{c_1}\rho_{c_1}.$
The converter $\pi^{\textnormal{e}_2}_{c_1}$ receives $U_{c_1}\rho_{c_1}$ at its inside interface. It requests a key $k_2$ from $\mathcal{K}$ and outputs $\tau=E_{k_2}U_{c_1}\rho_{c_1}$ through $\mathcal{C}^{\textnormal{q-insec}}.$
 The converter $\pi^{\textnormal{d}_2}_{c_2}$ receives the key $k_2$ from $\mathcal{K}$ and a (possibly) tampered state $\tau'$ and applies verification map $D_{k_2}$. It outputs the state $U_{c_1}\rho_{c_1}$ if verification succeeded, or $\ket{\text{err}}\bra{\text{err}}$ otherwise.
If verification succeeded, the converter $\pi_{c_2}^{\textnormal{u}_2}$ receives $U_{c_1}\rho_{c_1}$ at its inside interface and $U_{c_2}$ at its outside interface and outputs $U_{c_2}U_{c_1}\rho_{c_1}$. $\pi^{\textnormal{e}_3}_{c_2}$ requests a key $k_3$ from $\mathcal{K}$, applies $E_{k_3}$ and outputs $E_{k_3}U_{c_2}U_{c_1}\rho_{c_1}.$
$\pi^{\textnormal{d}_3}_{c_2}$ receives a possibly tampered state and the key $k_3$ from $\mathcal{K}.$ It outputs $\varrho_{c_2}=U_{c_2}U_{c_1}\rho_{c_1}$ if verification passes, otherwise it outputs $\ket{\text{err}}\bra{\text{err}}$.
 Based on Eq.(\ref{sim-auth}), we have:

\begin{eqnarray*}
  \pi_{c_1}^{\textnormal{u}_1}\pi_{c_2}^{\textnormal{u}_2}\pi_{c_1}^{\textnormal{e}_1}\pi_{c_2}^{\textnormal{d}_1} \pi^{\textnormal{e}_2}_{c_1}\pi^{\textnormal{d}_2}_{c_2}\pi^{\textnormal{e}_3}_{c_1}\pi^{\textnormal{d}_3}_{c_2}(\mathcal{K}||\mathcal{C}^{\textnormal{q-insec}}||\mathcal{C}^{\textnormal{c-auth}})\underset{3\epsilon^{\textnormal{q-sec}}}{\approx}\\
    \pi_{c_1}^{\textnormal{u}_1}\pi_{c_2}^{\textnormal{u}_2} \mathcal{C}^{\textnormal{q-sec}}\sigma^{\textnormal{q-sec}}_s\mathcal{C}^{\textnormal{q-sec}}\sigma^{\textnormal{q-sec}}_s\mathcal{C}^{\textnormal{q-sec}}\sigma^{\textnormal{q-sec}}_s.
\end{eqnarray*}
Now we have:
 \begin{equation}
\pi_{c_1}^{\textnormal{u}_1}\pi_{c_2}^{\textnormal{u}_2} \mathcal{C}^{\textnormal{q-sec}}\sigma^{\textnormal{q-sec}}_s\mathcal{C}^{\textnormal{q-sec}}\sigma^{\textnormal{q-sec}}_s\mathcal{C}^{\textnormal{q-sec}}\sigma^{\textnormal{q-sec}}_s= \mathcal{S}'\sigma^{\textnormal{q'-sec}}_s\sigma^{\textnormal{q-sec}}_s\sigma^{\textnormal{q-sec}}_s,
 \label{eq:ideal-interghange}
 \end{equation}
 where $\mathcal{S}' $ is the ideal channel where $\pi_{c_1}^{\textnormal{u}_1}$ and $\pi_{c_2}^{\textnormal{u}_2}$ have been merged with it. $\sigma^{\textnormal{q'-sec}}_s$ behaves the same as $\sigma^{\textnormal{q-sec}}_s,$ but it also outputs the leak of both communications $\ell^{\textnormal{q-sec}}_1$ and $\ell^{\textnormal{q-sec}}_2,$
and waits to check if the messages have been sent or not. If so, $\sigma^{\textnormal{q'-sec}}$ tells $\mathcal{S}'$ to output the correct value, otherwise $\ket{\text{err}}\bra{\text{err}}$ at the $C_2$ interface.
We move forward to the real case where the parties do not have quantum capabilities and delegate the computations (except $E_{k_1}$ and $D_{k_3}$) to the resource $S^{\textnormal{bb}}.$ With the similar analyses of the security proof in the previous section, we have:
\begin{equation}
\pi_{c_1}^{\textnormal{e}_1}\pi_{c_2}^{\textnormal{d}_3}(\mathcal{K}||\mathcal{C}^{\textnormal{q-insec}}||\mathcal{C}^{\textnormal{c-auth}}||\mathcal{S}^\textnormal{bb}_1 ||\mathcal{S}^{\textnormal{bb}}_2)\underset{3\epsilon^{\textnormal{q-sec}}}{\approx} \mathcal{S}^{\textnormal{2-bv}}\sigma^{\textnormal{q''-sec}}_s
\sigma^{\textnormal{q-sec}}_s\sigma^{\textnormal{q-sec}}_s.
\label{eq:S-bv}
\end{equation}
The simulator $\sigma^{\textnormal{q"-sec}}_s$ behaves the same as $\sigma^{\textnormal{q'-sec}}_s,$ but it also outputs the leak of both computations $\ell^{\textnormal{bb}}_1$ and $\ell^{\textnormal{bb}}_2$
and waits for the server to say if it does the correct computation or not. If so, $\sigma^{\textnormal{q"-sec}}$ tells $\mathcal{S}^\textnormal{2-bv}$ to output the correct value, otherwise $\ket{\text{err}}\bra{\text{err}}$ at the $C_2$ interface.
Now, plugging simulators $\sigma^{\textnormal{bb}}$ on the real and ideal systems does not increase their distance:
\begin{equation}
\pi_{c_1}^{\textnormal{e}_1}\pi_{c_2}^{\textnormal{d}_3}(\mathcal{K}||\mathcal{C}^{\textnormal{q-insec}}||\mathcal{C}^{\textnormal{c-auth}}||\mathcal{S}^\textnormal{bb}_1\sigma^{\textnormal{bb}}_s ||\mathcal{S}^{\textnormal{bb}}_2\sigma^{\textnormal{bb}}_s)\underset{3\epsilon^{\textnormal{q-sec}}}{\approx}\mathcal{S}^{2\textnormal{-bv}}\sigma^{\textnormal{q''-sec}}_s\sigma^{\textnormal{q-sec}}_s\sigma^{\textnormal{q-sec}}_s\sigma^{\textnormal{bb}}_s\sigma^{\textnormal{bb}}_s
\label{Eq:8}
\end{equation}
The final step is to replace $\mathcal{S}^\textnormal{bb}\sigma^{\textnormal{bb}}$ with the real BB-DQC protocols,
$\pi^{\textnormal{bb}}=(\pi_c^{\textnormal{bb}},\pi_s^{\textnormal{bb}})$:
 \[\pi_{c_1}^{\textnormal{e}_1}\pi_{c_{1}}^\textnormal{bb}\pi_{c_{2}}^\textnormal{bb}\pi_{c_2}^{\textnormal{d}_3}(\mathcal{K}||\mathcal{C}^{\textnormal{q-insec}}||\mathcal{C}^{\textnormal{c-auth}}||\mathcal{R}_1||\mathcal{R}_2)\underset{3\epsilon^{\textnormal{q-sec}}+2\epsilon^{\textnormal{bb}}}{\approx} \mathcal{S}^\textnormal{2-bv}\sigma^{\textnormal{q''-sec}}_s\sigma^{\textnormal{q-sec}}_s\sigma^{\textnormal{q-sec}}_s\sigma^{\textnormal{bb}}_s\sigma^{\textnormal{bb}}_s.\]

\begin{algorithm}[!h]
\caption{Modified Blind Verifiable Two-Client DQC for the Setting in Figure~\ref{fig:SampleMBQC} }
\begin{algorithmic}[1]
\STATE Let $\rho_{c_i}$ and $U_{c_i}$ be the input and the local unitary. 
 of $C_i$, respectively, where $i\in\{1,2\}$
\STATE Let $(E_{k},D_{k})$ be a pair of operations for authenticating.
\STATE $\text{Client}_1$ and the server run BB1-DQC($\rho_{c_1},U_{c_1}$).
\IF {error is detected}
\STATE $\text{Client}_1$ sends $e=1$ to $\text{Client}_2$ through an authenticated classical channel and both outputs $\ket{\text{err}}\bra{\text{err}}$ and abort the protocol.
\ELSE
\STATE the protocol outputs $E_kU_1 \rho_{c_1}$ and $k.$
\STATE $\text{Client}_1$ sends $k$ to $\text{Client}_2.$
\ENDIF
\STATE $\text{Client}_2$ and the server run BB2-DQC($E_kU_1 \rho_{c_1},U_2$) which outputs $\varrho_2.$
\IF {error is detected}
\STATE $\text{Client}_1$ and $\text{Client}_2$ output $\ket{\text{err}}\bra{\text{err}}$
\ELSE
\STATE $\text{Client}_2$ outputs $\varrho_{c_2}=U_2U_1\rho_{c_1}.$
\ENDIF

\end{algorithmic}
\end{algorithm}
So the composition of the authentication protocol and the BV-DQC protocols with the particular structure described above, constructs the 2-client BV-DQC resource, $\mathcal{S}^\textnormal{2-bv}$ with error equal to $2\epsilon^{\textnormal{bb}}+3\epsilon^{\textnormal{q-sec}}$. The security proof can be directly extended to the case of a generic multi party DQC. In this case the total error, $\epsilon^{\textnormal{n-bv}}$ is the sum of the errors of all runs of BV-DQC and QAS. In the case of fully connected network, with $n$ clients and each with $m$ unitaries, it is equal to:
\[\epsilon^{\textnormal{n-bv}}=(mn)\epsilon^{\textnormal{bb}}+(n(n-1)(m-1)+2n)\epsilon^{\textnormal{q-sec}}. \]
$\;\; \;\;\;\;\;\;\;\;\; \;\; \;\;\;\;\;\;\;  \;\;\;\;\;\;\;\;\;\;\;\; \;\;\;\;\;\;\;\;\;\;\;\; \;\;\;\;\; \;\;\;\;\;\;\;\;\;\;\;\; \;\;\;\;\;  \;\;\;\;\; \;\;\;\;\;\;\;\;\;\;\;\; \;\;\;\;\;\;\;\;\;\; \;\;\;\;\;\;\;\;\;\;\;\; \;\;\;\;\;\;\;\;\;\;\;\;\;\; \;\;\blacksquare$

\end{proof}

\section{Conclusion}
\label{sec:conclu}
In this paper, we studied the multi-client DQC problem. In our setting, we considered the global unitary to be made up of local unitaries that are individually decided by the clients. Each client's part is kept secret from the server and the other clients (assuming honest client behaviour). We constructed a verifiable composable secure multi-client DQC scheme from any verifiable composable secure single-client DQC protocols and quantum authentication codes. We proved the security and then optimised the protocol by removing the quantum communication between the clients and the server.

\textit{Acknowledgements --} The authors thank Christopher Portman for useful discussions and Mark M.~Wilde for his comments on an earlier version of this manuscript. The authors acknowledge support from Singapore's Ministry of Education and National Research Foundation, and the Air Force Office of Scientific Research under AOARD grant FA2386-15-1-4082. This material is based on research funded in part by the Singapore National Research Foundation under NRF Award NRF-NRFF2013-01.

\bibliographystyle{splncs03}
\bibliography{apssamp}
\end{document}